\documentclass[%
 reprint,
superscriptaddress,
 amsmath,amssymb,
 aps,
pra,
]{revtex4-2}
\usepackage{graphicx}
\usepackage{dcolumn}
\usepackage{bm}
\usepackage{amsmath}
\usepackage{amsthm}
\usepackage{subfigure}
\usepackage{xcolor}
\usepackage{ulem}
\usepackage{pifont}

\usepackage{graphicx,accents}

\makeatletter
\DeclareRobustCommand{\cev}[1]{%
  \mathpalette\do@cev{#1}%
}
\newcommand{\do@cev}[2]{%
  \fix@cev{#1}{+}%
  \reflectbox{$\m@th#1\vec{\reflectbox{$\fix@cev{#1}{-}\m@th#1#2\fix@cev{#1}{+}$}}$}%
  \fix@cev{#1}{-}%
}
\newcommand{\fix@cev}[2]{%
  \ifx#1\displaystyle
    \mkern#23mu
  \else
    \ifx#1\textstyle
      \mkern#23mu
    \else
      \ifx#1\scriptstyle
        \mkern#22mu
      \else
        \mkern#22mu
      \fi
    \fi
  \fi
}

\makeatother

\newtheorem{thm}{Theorem}

\newtheorem{prop}[thm]{Proposition}
\newtheorem{conj}[thm]{Conjecture}

\newcommand{\N}{\mathcal{N}}

\begin{document}
\preprint{APS/123-QED}

\title{General Communication Enhancement via the Quantum Switch}

\date{\today}


\author{Zhen Wu}
\affiliation{School of Mathematical Sciences, Shanghai Jiao Tong University, 800 Dongchuan Road, Shanghai, China}

\author{James Fullwood}
\email{fullwood@hainanu.edu.cn}
\affiliation{School of Mathematics and Statistics, Hainan University, 58 Renmin Avenue, Haikou, China}

\author{Zhihao Ma}
\email{mazhihao@sjtu.edu.cn}
\affiliation{School of Mathematical Sciences, Shanghai Jiao Tong University, 800 Dongchuan Road, Shanghai, China}

\author{Siqi Zhou}
\affiliation{School of Mathematical Sciences, Shanghai Jiao Tong University, 800 Dongchuan Road, Shanghai, China}

\author{Qi Zhao}
\email{zhaoqithu10@gmail.com}
\affiliation{QICI Quantum Information and Computation Initiative, Department of Computer Science, The University of Hong Kong, Pok Fu Lam Road, Hong Kong}

\author{Giulio Chiribella}
\email{giulio@cs.hku.hk}
\affiliation{QICI Quantum Information and Computation Initiative, Department of Computer Science, The University of Hong Kong, Pok Fu Lam Road, Hong Kong}
\affiliation{Department of Computer Science, University of Oxford, Wolfson Building, Parks Road, Oxford, UK}
\affiliation{Perimeter Institute for Theoretical Physics, 31 Caroline Street North, Waterloo,  Ontario, Canada}

\begin{abstract}
	    Recent studies have shown that quantum information may be effectively transmitted by a finite collection of completely depolarizing channels in a coherent superposition of different orders, via an operation known as the quantum $\tt SWITCH$. Such results are quite remarkable, as completely depolarizing channels taken in isolation and in a definite order can only output white noise. For general channels however, little is known about the potential communication enhancement provided by the quantum $\tt SWITCH$. In this Letter, we define an easily computable quantity $\mathcal{P}_n$ associated with the quantum ${\tt SWITCH}$ of $n$ copies of a fixed channel, and we conjecture that $\mathcal{P}_n>0$ is both a necessary and sufficient condition for communication enhancement via the quantum $\tt SWITCH$. In support of our conjecture, we derive a simple analytic expression for the classical capacity of the quantum $\tt SWITCH$ of $n$ copies of an arbitrary Pauli channel in terms of the quantity $\mathcal{P}_n$, which we then use to show that our conjecture indeed holds in the space of all Pauli channels. Utilizing such results, we then formulate a communication protocol involving the quantum $\tt SWITCH$ which enhances the private capacity of the BB84 channel. 
\end{abstract}

\maketitle

{\em Introduction.~}Quantum Shannon theory---the extension of Shannon's communication theory to the quantum domain~\cite{Wilde}---has achieved communication advantages not possible with the classical theory, such as secure quantum key distribution and distributed quantum computation   \cite{BB84,E91,quantuminternet}. Recently, there has been interest in a further extension of quantum Shannon theory where quantum systems are not only used to carry information, but also to control the configuration of the communication devices~\cite{IncomPRA,Comm_Phys,QuanResPRA,PRSA}. In particular, it has been observed that the ability to combine  quantum devices in a coherent superposition of different orders can give rise  to   advantages over the standard scenario where the corresponding channels are used in parallel or in a sequence~\cite{IncomPRL,QSWofCDC,CSwitchPRL,CSwitchPRA,QSwitchPRL1,QSwitchNJP,PhysScri,Liu_2023,Ther_Liu,Entropy,JSAC,QuanCommPRR,IncreCommPRR}.  These advantages are based on a specific instance of indefinite causal order known as the quantum ${\tt SWITCH}$~\cite{Switch2},  an operation whose input is a finite collection of quantum channels and whose output is a new quantum channel.  In the quantum ${\tt SWITCH}$, the input channels are executed in a superposition of causal orders controlled by an auxiliary quantum system. Besides quantum communication,  it has been shown that the quantum ${\tt SWITCH}$ yields advantages in various tasks, such as charging quantum batteries~\cite{QuanBattery,ChargeBatteryPRR}, distinguishing quantum processes~\cite{ChannelDiscrimination,UnitaryJMP,QIP,PerfectDisPRA,IJQI},  reducing quantum communication complexity~\cite{CommComple}, and improving the precision of quantum metrology~\cite{MetrologyPRL,MetrologyNP,MetrologyPRA}.

Notably, it has been observed that the quantum ${\tt SWITCH}$ of two completely depolarizing channels has the capacity to transmit classical information~\cite{CSwitchPRL}, despite the fact that any combination of completely depolarizing channels in a definite configuration is only able to transmit white noise. More generally, it has been shown that the quantum ${\tt SWITCH}$ of a finite collection of completely depolarizing channels can even transmit \textit{quantum} information when arranged in a superposition of cyclic orders~\cite{QSWofCDC}, and some entanglement-breaking channels can achieve perfect communication via the quantum ${\tt SWITCH}$~\cite{QSwitchNJP,PhysScri}.

Although many studies have  illustrated the capacity enhancement of the quantum ${\tt SWITCH}$, such results however are limited to specific examples, and little is known about the communication enhancement provided by the quantum ${\tt SWITCH}$ for generic channels. Due to the fact that in general there is no analytic expression for the quantum/classical capacity of quantum channels, novel approaches must be developed to gain insight into the general problem of characterizing the  communication enhancements offered by the quantum $\tt SWITCH$. A clue for such an approach appears in Ref.~\cite{IncomPRA}, where the quantum ${\tt SWITCH}$  was applied to two copies of a  randomly generated a mixture $\mathcal{C}$ of unitary qubit channels, thus generating a new  quantum  channel $\mathcal{S}(\mathcal{C},\mathcal{C})$. Numerical analysis 
suggested that the Holevo information of the output channel $\mathcal{S}(\mathcal{C},\mathcal{C})$ may be correlated with a real-valued measure
of non-commutativity of the Kraus operators of the original channel $\mathcal{C}$. However, no theoretical explanation for this possible link between Holevo information and non-commutativity has been found thus far. Moreover, the Holevo information is only a lower bound for the classical capacity, and therefore an increase of the  Holevo information does not necessarily imply an increase in the classical capacity. 

In this Letter, we establish precise mathematical results which lay the foundation for a systematic study of communication enhancement via the quantum $\tt SWITCH$. In particular, we define an easily computable quantity $\mathcal{P}_n$ associated with the quantum $\tt SWITCH$ of $n$ channels, and we conjecture that $\mathcal{P}_n>0$ is a necessary and sufficient condition for the quantum $\tt SWITCH$ to increase capacity (outside of a set of measure-zero). In support of this conjecture, we consider the three-dimensional simplex of all Pauli channels (minus the point corresponding to the completely depolarizing channel when $n$ is odd), and prove that the quantum ${\tt SWITCH}$ of the forward and backward orders of the the $n$-fold composition of an arbitrary Pauli channel enhances the classical capacity if and only if $\mathcal{P}_n>0$. In the case that $n$ is even, we show $\mathcal{P}_n=0$ if and only if the Kraus operators of the associated Pauli channel pair-wise commute, thus the link between capacity enhancement provided by the quantum ${\tt SWITCH}$ and non-commutativity as suggested in~\cite{IncomPRA} is made mathematically precise by our results. 

We also show that in the simplex of all Pauli channels, the locus $\mathcal{P}_n=0$ is contained in the edges of the simplex, thus the superposition of the forward and backward orders of $n$ compositions of a Pauli channel as implemented by the quantum ${\tt SWITCH}$ enhances classical capacity outside a set of measure-zero. In addition to the case of Pauli channels, we consider a family of qudit depolarizing channels (of arbitrary dimension) parametrized by an interval, and prove that $\mathcal{P}_n>0$ is  again a necessary and sufficient condition for capacity enhancement via the quantum $\tt SWITCH$ of the forward and backward orders of the $n$-fold composition of a fixed channel in the family. Finally, we provide a potential application to quantum cryptography, by showing that the quantum $\tt SWITCH$ increases the private capacity of a general BB84 channel.

{\em The quantity $\mathcal{P}_n$ and causal gains.~}In quantum Shannon theory, a communication process is mathematically described by a quantum channel, i.e., a completely positive trace-preserving (CPTP) linear map $\mathcal{C}:\text{Lin}(\mathcal{H}_A)\to \text{Lin}(\mathcal{H}_B)$, where $\mathcal{H}_A$ and $\mathcal{H}_B$ are the Hilbert spaces of the input and output systems of the channel, $\text{Lin}(\mathcal{H})$ contains all linear operators of $\mathcal{H}$. Any such  map admits a Kraus representation $\mathcal{C}(\rho) = \sum_i C_i\rho C_i^\dagger$, where the Kraus operators $\{C_i\}$ satisfy $\sum_i C_i^\dagger C_i = I$. 

The quantum ${\tt SWITCH}$ is a \emph{supermap}~\cite{SwitchWD,Switch2} which takes as its input a finite collection of quantum channels and outputs a channel that uses the input channels in a superposition of orders which is entangled with an auxiliary control system~\cite{Switch2,SwitchN}. In particular, given a collection of CPTP maps $\{\mathcal{C}^i:\text{Lin}(\mathcal{H}_A)\to \text{Lin}(\mathcal{H}_B)\}_{i=1}^n$, then the quantum ${\tt SWITCH}$ of the channels $\{\mathcal{C}^{i}\}_{i=1}^{n}$ with respect to a subset $\mathbf{S}$ of permutations on $n$ letters is a channel of the form $\mathcal{S}(\mathcal{C}^{1},\ldots,\mathcal{C}^{n}):\text{Lin}(\mathcal{H}_A)\otimes \text{Lin}(\mathcal{H}_C)\to \text{Lin}(\mathcal{H}_B)\otimes \text{Lin}(\mathcal{H}_C)$, where $\mathcal{H}_C$ is the Hilbert space of the control system. If $\{C^i_{s_i}\}$ is a collection of Kraus operators for the channel $\mathcal{C}^i$, then fixing a control state $\omega \in \mathcal{H}_C$ yields an effective channel $\mathcal{S}^n:\text{Lin}(\mathcal{H}_A)\to \text{Lin}(\mathcal{H}_B)\otimes \text{Lin}(\mathcal{H}_C)$, which may be written as~\cite{QSWofCDC}
\begin{equation}\label{switchN}
	\mathcal{S}^n(\rho) = \sum_{\pi^k,\pi^l\in \mathbf{S}} C_{\pi^k \pi^l}(\rho) \otimes \omega_{k l}|k\rangle\langle l|_C,
\end{equation}
where $\mathbf{S}$ is the given subset of permutations,
$
	C_{\pi^k \pi^l}(\rho) := \sum_{s_1,\dots,s_n} C_{\pi^k(s_1,\dots,s_n)} \rho C_{\pi^l(s_1,\dots,s_n)}^\dagger
$, and $C_{\pi^k(s_1,\dots,s_n)}:= C^{\pi^k(1)}_{s_{\pi^k(1)}}\cdots C^{\pi^k(n)}_{s_{\pi^k(n)}}$.
We note that while we have written the effective channel $\mathcal{S}^n$ in terms of Kraus operators for the channels $\{\mathcal{C}^i\}$, the quantum ${\tt SWITCH}$ is in fact independent of a choice of Kraus operators for its input channels.

Given such data defining the effective channel $\mathcal{S}^n$, we define the quantity~\cite{footnote1}
\begin{equation}\label{generalPn}
\mathcal{P}_{n}=
1-\frac{1}{m^2} \min_{\rho} \sum_{\pi^k,\pi^l\in \mathbf{S}} \text{Tr}\bigg(C_{\pi^k \pi^l}(\rho)\bigg)\, ,
\end{equation}
where $m:=|\mathbf{S}|$. For a given control state of the form $\omega = |\omega\rangle\langle \omega|$ with $|\omega\rangle = \sum_{i\in  \mathbf S} |i\rangle /\sqrt{m}$, $\mathcal{P}_{n}$ may be viewed as the maximum probability (as $\rho$ varies over all input states) of obtaining the measurement outcome $F_2$ associated with the projective measurement $\{F_1 = |\omega\rangle\langle\omega|, F_2 = I-|\omega\rangle\langle\omega|\}$. In the case $n=m=2$, the quantity $\mathcal{P}_{n}$ was used to characterize the non-commutativity of generic channels~\cite{IncomPRA,IncomPRL}, which is closely related to measurement incompatibility~\cite{MeasurementInc}. For general $n$ and $\mathbf{S}$, we show in the Appendix that $\mathcal{P}_{n}$ is \emph{not} necessarily a measure of non-commutativity of channels, as we prove $\mathcal{P}_{n}$ is zero if and only if the product of Kraus operators is $\mathbf{S}$-invariant, i.e., for all $s_1,\dots,s_n$ and $\pi^k,\pi^l\in \mathbf{S}$, we have $C_{\pi^k(s_1,\dots,s_n)} = C_{\pi^l(s_1,\dots,s_n)}$. However, our main interest in the quantity $\mathcal{P}_n$ lies in its connection to communication enhancement via the quantum ${\tt SWITCH}$, which we now explain.

Consider the case where $\mathcal{C}^i=\N$ for $i=1,...,n$, where $\mathcal{N}:\text{Lin}(\mathcal{H}_A)\to \text{Lin}(\mathcal{H}_B)$ is a fixed quantum channel, and let $\mathcal{N}^n$ denote the $n$-fold composition $\N\circ \cdots \circ \N$ (which may be obtained from the effective channel $\mathcal{S}^n$ by tracing out the control system).   For every  capacity measure $f$ satisfying the data-processing inequality $f  ({\cal A} \circ {\cal B} ) \le  f  (\cal B)$, for arbitrary quantum channels $\cal A$ and $\cal B$,  one  then obtains the \textit{bottleneck inequality}  
\begin{equation}\label{bottleneckineq}
	f(\mathcal{N}^n) \leq f(\mathcal{S}^n)\, .
\end{equation}
The classical, quantum, and private capacity, as well as the  Holevo information and coherent information are all capacity measures satisfying the above inequality.  Hence, we define the associated \textit{causal gain} to be the non-negative quantity $\delta_f$ given by \begin{equation}\label{CausalGain}
	 \delta_f = f(\mathcal{S}^n) - f(\mathcal{N}^n).
\end{equation}
The causal gain is a direct measure of the communication enhancement of the channel $\N$ which is achieved by inputting $n$-copies of the channel into the quantum $\tt SWITCH$.

In the appendix we show $\mathcal{P}_n=0$ implies $\delta_f=0$, so that $\mathcal{P}_n>0$ is a necessary condition for $\delta_f>0$. A natural question is whether or not $\mathcal{P}_n>0$ is also a \emph{sufficient} condition for positive causal gain. Since the only known channels for which $\mathcal{P}_n>0$ does not imply $\delta_f>0$ are completely dephasing channels composed with a unitary gate or discard-preparation channels, we put forth the following: 
\begin{conj}\label{conj}
For all channels outside a set of measure-zero, the condition $\mathcal{P}_n>0$ is necessary and sufficient for $\delta_f>0$.
\end{conj}

As the set of channels satisfying $\mathcal{P}_n = 0$ forms a measure-zero set for any permutation set $\mathbf{S}$ with $|\mathbf{S}|>1$, establishing the validity of Conjecture~\ref{conj} would imply that the quantum $\tt SWITCH$ enhances communication almost surely. However, due to superadditivity of Holevo information and coherent information, it is quite challenging to determine the classical and quantum capacities for a general quantum channel, thus making Conjecture~\ref{conj} quite difficult to prove in full generality. In fact, even evaluating Holevo or coherent information for generic channels is also challenging, requiring a non-convex optimization over all input states. Moreover, the computation of Holevo information for a generic channel is actually an NP-complete problem, even for entanglement breaking channels~\cite{HolevoNPC}. Nevertheless, in what follows we establish some substantial results in support of Conjecture~\ref{conj}.

{\em Pauli channels with forward and backward orders.~}
In this section, we establish the validity of Conjecture~\ref{conj} when restricted to the set of Pauli channels (i.e., for channels $\N$ of the form $\mathcal{N}(\rho)=\sum_{i=0}^3 p_i \sigma_i\rho\sigma_i$ with $(p_0,\dots,p_3)$ a probability vector), and for $\mathbf{S}$ the permutation set consisting of the identity permutation $(1,\ldots,n)$ and its reversal $(n,\ldots,1)$. In such a case, the quantum $\tt SWITCH$ places $n$ copies of a Pauli channel $\N$ in a superposition of the forward and backward orders $\mathcal{N}_n\circ\cdots\circ\mathcal{N}_1$ and $\mathcal{N}_1\circ\cdots\circ\mathcal{N}_n$ (where $\N_i=\N$ for all $i$), as illustrated in Fig.~\ref{Fig1}. While when used individually the forward and backward orders are indistinguishable, we show that when such orders are placed in a superposition via the quantum ${\tt SWITCH}$, an enhancement of classical capacity and coherent information occurs almost surely.

Since we consider a setting where the permutation set $\mathbf{S}$ only contains two elements, we fix the control state to be $\omega = |+\rangle\langle+|$, and then apply the projective  measurement $\{F_1 = |\omega\rangle\langle\omega|, F_2 = I-|\omega\rangle\langle\omega|\}$. We then find that the probability of obtaining $F_2$ is in fact independent of the initial state $\rho$, and therefore  it is equal to the quantity $\mathcal{P}_n$ (as given by \eqref{generalPn}), so that for all states $\rho$,
\begin{equation}\label{NC}
	\begin{aligned}
		\text{Tr}\bigg( (I\otimes F_2)\ \mathcal{S}^n(\rho) \bigg) = \mathcal{P}_n \, .
	\end{aligned}
\end{equation}

\begin{figure}

	\includegraphics[width=0.5\textwidth]{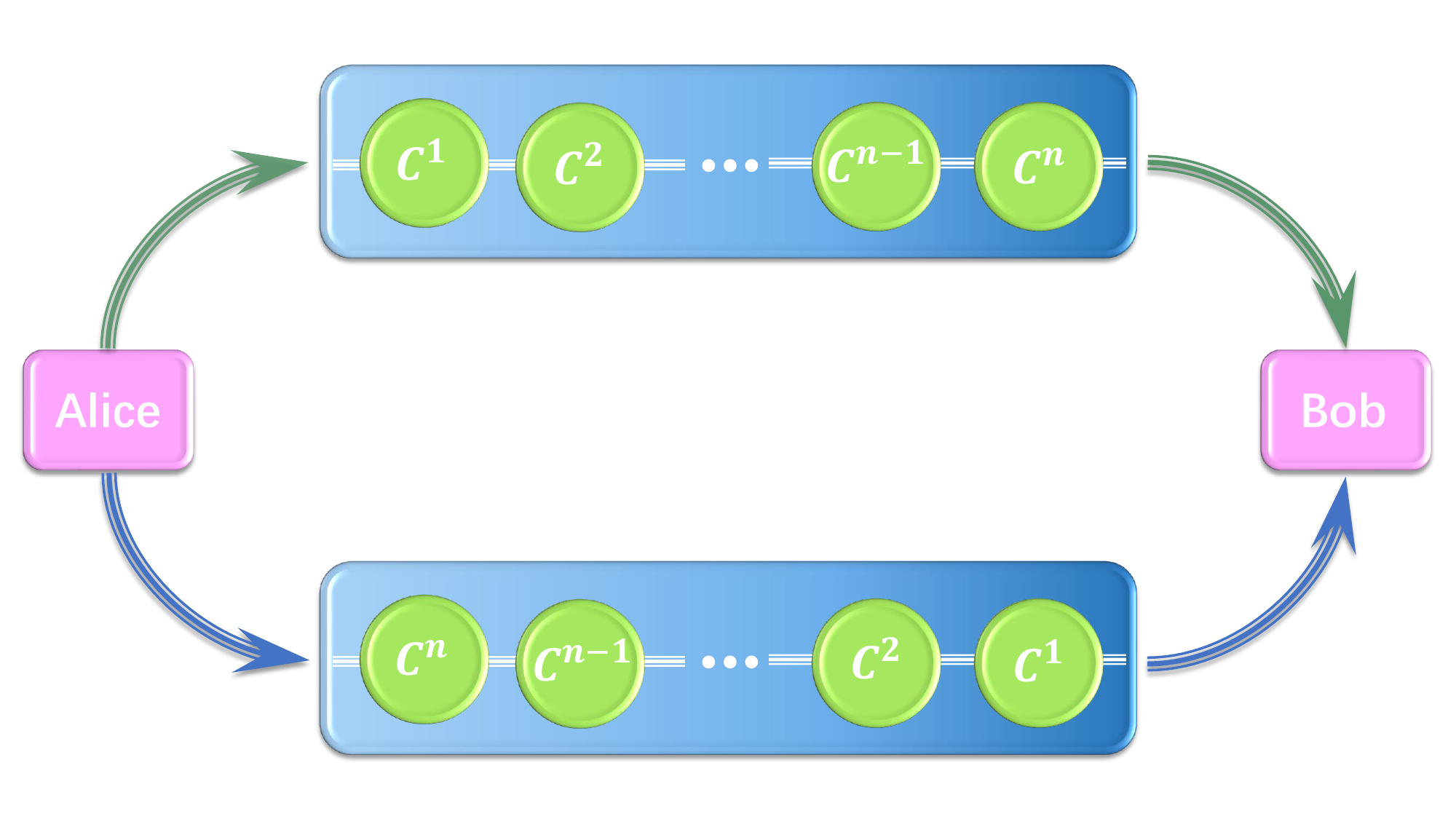}
	\caption{{\label{Fig1}} Alice sends messages to Bob through $n$ copies of a Pauli channel $\mathcal{C}^1,\dots,\mathcal{C}^n$. These channels are then placed in a superposition forward and backward orders as implemented by the quantum $\tt SWITCH$. }
\end{figure}

When $n$ is even, $\mathcal{P}_n$ is a faithful measure of non-commutativity, i.e., $\mathcal{P}_n = 0$ if and only if the Kraus operators of $\mathcal{N}$ pairwise-commute. More importantly, for generic $n$, $\mathcal{P}_n = 0$ happens only if the Choi rank of $\mathcal{N}$ is no more than 2, and such Pauli channels have been shown either to be degradable or anti-degradable~\cite{degradable,Structure}. Moreover, as the set of Pauli channels is parametrized by the 3-dimensional simplex of probability vectors in $\mathbb{R}^4$, and since the subset satisfying $\mathcal{P}_n = 0$ is contained in the edges of this simplex (which is a subset of measure-zero), it follows that $\mathcal{P}_n >0$ almost surely.

In regards to the establishment Conjecture~\ref{conj} in the context at hand, we now relate the quantity $\mathcal{P}_n$ to the causal gain $\delta_f$ when $f$ is either the classical capacity $C$ or the coherent information $I_c$. For this, we first note that the independence of equation \eqref{NC} on the state $\rho$ implies that  the effective channel $\mathcal{S}^n$ is given by
\begin{equation}\label{EffectiveChannel}
	\mathcal{S}^n(\rho) = (1-\mathcal{P}_n) \Phi_+ (\rho) \otimes |+\rangle\langle+| + \mathcal{P}_n \Phi_-(\rho) \otimes |-\rangle\langle-|,
\end{equation}
where $\Phi_+(\rho) = \sum_i s_i \sigma_i\rho \sigma_i$ and $\Phi_-(\rho)=\sum_i t_i \sigma_i\rho \sigma_i$ are two Pauli channels. Moreover, since $\mathcal{N}^n = \sum_{i=0}^3 q_i \sigma_i\rho\sigma_i$ is the composite of $n$ copies of $\mathcal{N}$, it follows that $\mathcal{N}^n = (1-\mathcal{P}_n) \Phi_+ + \mathcal{P}_n \Phi_-$. In the Appendix, we then prove the following formula for the classical capacity $C$ of the effective channel $\mathcal{S}^n$: 
\begin{thm}\label{thm1}
Let $\Phi_{\pm}$ be as in equation \eqref{EffectiveChannel}. Then 
	\begin{equation}\label{classicalcapa}
		C(\mathcal{S}^n) = (1-\mathcal{P}_n) C(\Phi_+) + \mathcal{P}_n C(\Phi_-).
	\end{equation}
\end{thm}

As the classical capacity of a Pauli channel is equal to its Holevo information~\cite{Additivity}, and since the Holevo information may be computed in terms of the minimal entropy via the formula $\chi=1-H^{\text{min}}$~\cite{HolevoPauli3}, we can then use Theorem~\ref{thm1} to derive an explicit expression for the classical causal gain $\delta_C$. More precisely, given a Pauli channel $\mathcal{M}$ with eigenvalues $\{\lambda_i\}$, corresponding to the solution set of  the characteristic equation $\mathcal{M} (A) = \lambda A$, the minimal entropy of $\mathcal{M}$ is given by $H^{\text{min}}(\mathcal{M})=h(\lambda)$, where $h(x) = H(\frac{1+x}{2})$, $H$ is the binary entropy and $\lambda = \max_i |\lambda_i|$. Therefore, setting $\lambda$, $\mu$ and $\nu$ to be the maximum eigenvalues of $\N^n$, $\Phi_{+}$ and $\Phi_{-}$ respectively, the classical capacity of the effective channel is $C(\mathcal{S}^n) = 1-(1-\mathcal{P}_n)h(\mu) - \mathcal{P}_n h(\nu)$, so that the classical causal gain may be written as
\begin{equation} \label{DELTAC17}
    \delta_C = h(\lambda)-(1-\mathcal{P}_n)h(\mu) - \mathcal{P}_n h(\nu)\, .
\end{equation}
Utilizing Eq.~\eqref{DELTAC17}, one may show that Conjecture~\ref{conj} indeed holds for classical causal gains of Pauli channels with forward and backward orders.

Contrary to classical capacity, the quantum capacity of Pauli channels is still unknown since coherent information is not additive, even for qubit channels. Hence, we consider the causal gain associated with coherent information $I_c$~\cite{Schumacher}. For a Pauli channel $\mathcal{N}$, the coherent information attains a maximum on the completely mixed state $\frac{I}{2}$, which is called the \textit{hashing bound}~\cite{Wilde,Hashing}: $I_c(\mathcal{N}) = 1 - H(\vec{p})$, where $\vec{p} = (p_0,p_1,p_2,p_3)$ is the probability vector associated with the Pauli channel $\mathcal{N}$.
On the other hand, since $\Phi_\pm$ are both Pauli channels, it follows from Eq.~\eqref{EffectiveChannel} that the coherent information of $\mathcal{S}^n$ also attains a maximum on $\frac{I}{2}$:
$
	I_c(\mathcal{S}^n) = 1- (1-\mathcal{P}_n)H(\vec{s}) -  \mathcal{P}_nH(\vec{t}).
$
As such, the causal gain $\delta_I$ associated with the coherent information takes a similar form to that of $\delta_C$, as
\begin{equation}\label{CICG}
	\delta_I = H(\vec{q}) - (1-\mathcal{P}_n)H(\vec{s}) -  \mathcal{P}_nH(\vec{t})\, .
\end{equation}
Similar to the case of the classical causal gain, it is then straightforward to deduce that $\mathcal{P}_n>0$ is a necessary and sufficient condition for $\delta_I>0$ almost surely, in support of Conjecture~\ref{conj}.

It is worth noting that the increase of coherent information does not necessarily imply the rise of quantum capacity, since quantum capacity describes the highest rate that quantum information can be transmitted over many uses of a channel. However, since quantum coherent information quantifies the information lost to the environment during quantum communication~\cite{SurveyCapa}, the condition $\delta_I>0$ implies that the quantum $\tt SWITCH$ will protect quantum communication. 
Moreover, as coherent information is directly related to the entanglement-assisted classical capacity~\cite{ECC}, and since the coherent information of $\mathcal{S}^n$ attains its maximum on $\frac{I}{2}$, one may also use equation \eqref{CICG} to deduce necessary and sufficient conditions for causal gains associated with entanglement-assisted classical capacity of channels.

In summary, the results established in this section may be used to prove the following theorem (whose detailed proof appears in the Appendix).

\begin{thm}\label{DeltaCX}
Let $f$ be the classical capacity or coherent information, and let $\delta_f$ be the causal gain associated with forward and backward orders. Then in the simplex of all Pauli channels, the condition $\mathcal{P}_n>0\iff \delta_f>0$ holds almost surely (i.e., outside a set of measure-zero). 
\end{thm}

We note than when $n$ is even, the almost surely quantifier may be removed from Theorem~\ref{DeltaCX}, as $\mathcal{P}_n>0 \iff\delta_f>0$ for all Pauli channels. In the case that $n$ is odd however, the only counter-example to the equivalence $\mathcal{P}_n>0 \iff\delta_f>0$ is the completely depolarizing channel.

{\em Beyond qubits: classical causal gains of depolarizing channels.~}
In this section, we establish the validity of Conjecture~\ref{conj} for a collection of qudit depolarizing channels parametrized by the unit interval $[0,1]$. In particular, given $p\in [0,1]$ and $d>1$, let $\mathcal{D}_p^{d}$ be the channel given by 
\[\mathcal{D}_p^d(\rho) = (1-p)\rho + p\text{Tr}(\rho)\frac{I}{d}\, .
\]
Inputting $n$ copies of $\mathcal{D}_p^{d}$ into the quantum $\tt SWITCH$ with respect to forward and backward orders, we find
\[
\mathcal{P}_n = \frac{1}{2}-\frac{1}{4}\bigg((d+1)(1-p+\frac{p}{d})^n - (d-1)(1-p-\frac{p}{d})^n\bigg)\, .
\]
Furthermore, in Appendix E we prove a simple formula for the classical capacity of the effective channel $\mathcal{S}^n$ in terms of the quantity $\mathcal{P}_n$, thus yielding an explicit expression for the classical causal gain $\delta_C$. We are then able to prove the following:
\begin{thm} \label{TMX87}
Let $\delta_C$ be the classical causal gain for qudit depolarizing channels with respect to forward and backward orders. Then the condition $\mathcal{P}_n>0\iff \delta_f>0$ holds almost surely (i.e., outside a set of measure-zero).
\end{thm}

Similar to Theorem~\ref{DeltaCX}, the almost surely quantifier in Theorem~\ref{TMX87} is needed to exclude the cases where $n$ is odd and $\mathcal{D}_p^d$ is completely depolarizing.

Somewhat counterintuitively, we find that the behavior of $\delta_C$ depends heavily on the level of noise determined by the value of $p$. For $p<0.5$, the classical causal gain benefits both from increasing $d$ and $n$, and in such a case, the optimal number $n_{opt}$ of channels for communication enhancement via quantum $\tt SWITCH$ approximately satisfies $p\cdot n_{opt} \approx 1.$ However, when the noise is slightly larger than 0.5, $\delta_C$ decreases with respect to $n$ for dimensions $d \leq 5$, while $\delta_C$ increases for $d>5$. Finally, as the noise approaches 1, classical causal gains tend to zero as $d$ gets large, which agrees the result about classical causal gains of completely depolarizing channel~\cite{CSwitchPRL} (see the corresponding figures in the Appendix).

{\em Noise reduction in quantum cryptography.~}
Quantum cryptography is a branch of quantum Shannon theory that uses quantum properties to achieve cryptographic tasks~\cite{QuanCryp} . In 1984, Charles Bennett and Gilles Brassard proposed the first quantum cryptography protocol---known as the \textit{BB84 protocol}---which is provably secure~\cite{Simpleproof,unconditionalsecurity}. Suppose a sender (Alice) and a receiver (Bob) use the so-called BB84 channel $\mathcal{N}_q$ given by
\[
\mathcal{N}_q(\rho) = (1-q)^2 \rho + q(1-q) X\rho X + q^2 Y\rho Y + q(1-q) Z\rho Z
\]
to communicate classical messages, whose private capacity is equal to the maximal achievable key rate with a quantum bit error rate of $q$ in the BB84 protocol~\cite{AdditiveExtension}.

As the private capacity of the BB84 channel is still unknown, we consider the following upper bound~\cite{AdditiveExtension}
\begin{equation}\label{BB84}
    C_p(\mathcal{N}_q) \leq H\bigg(\frac{1}{2} - 2q(1-q)\bigg) - H(2q(1-q))\, ,
\end{equation}
where $H$ is the binary entropy~\cite{footnote2}.

While $\mathcal{N}_q$ and $\mathcal{N}_{1-q}$ have the same private capacity due to the fact that they differ by a unitary transformation, 
the causal gains provided by the quantum $\tt SWITCH$ are in fact different for $\mathcal{N}_q$ and $\mathcal{N}_{1-q}$. This follows from the fact that the quantity $\mathcal{P}_n$ and causal gains are in general \textit{not} unitarily covariant, in the sense that the output channel of two copies $\mathcal{U}\circ \mathcal{C}$ via quantum $\tt SWITCH$ is not equal to $(\mathcal{U} \otimes \mathcal{I}) \circ \mathcal{S}^2(\mathcal{C})$ for general unitary channels $\mathcal{U}$. Based on this feature, we now formulate a new communication protocol which enhances the private capacity of the BB84 channel.

\begin{figure}
        \subfigure[Protocol]{\includegraphics[width=\columnwidth]{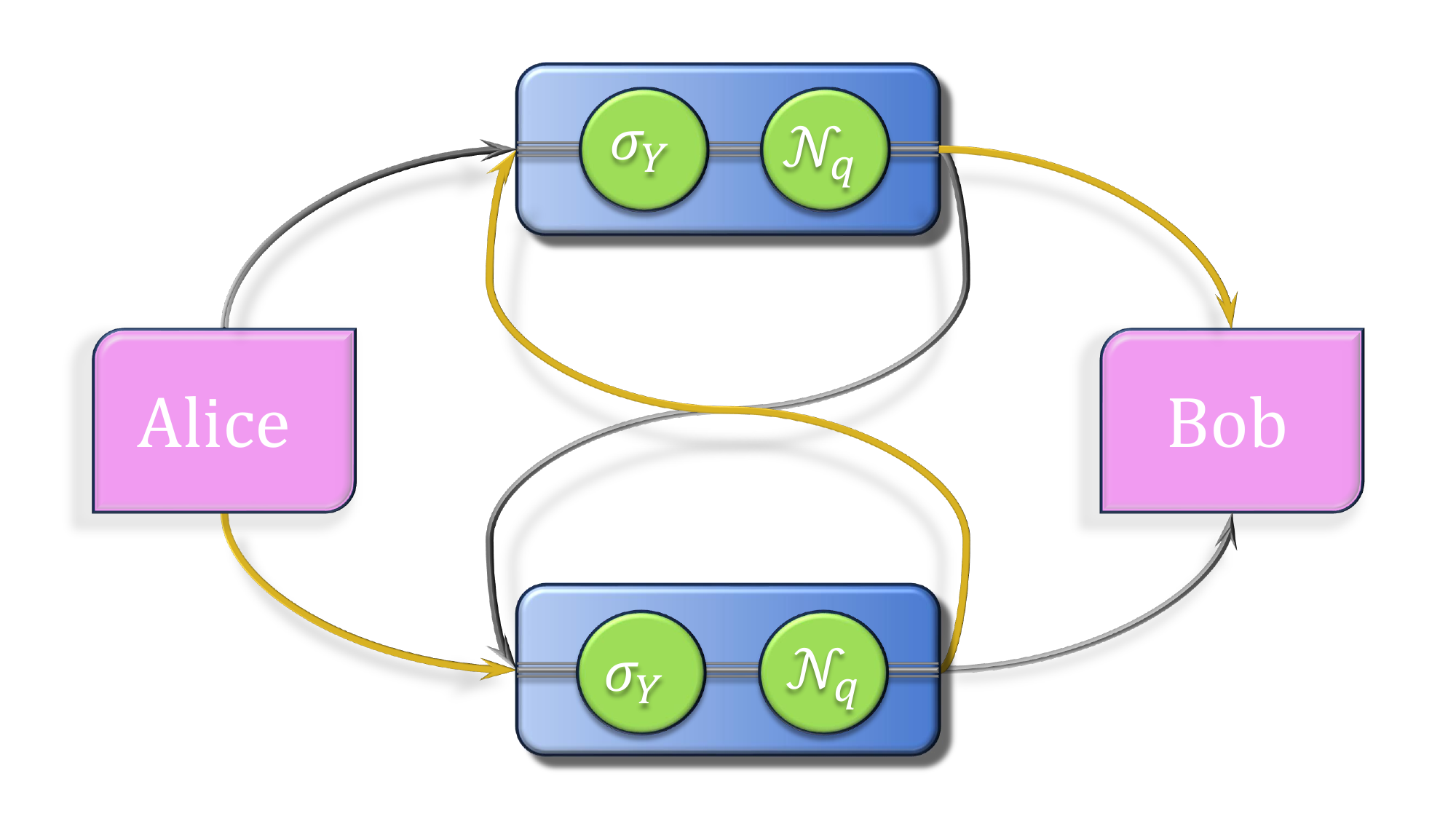}}
        \subfigure[Private capacity]{\includegraphics[width=\columnwidth]{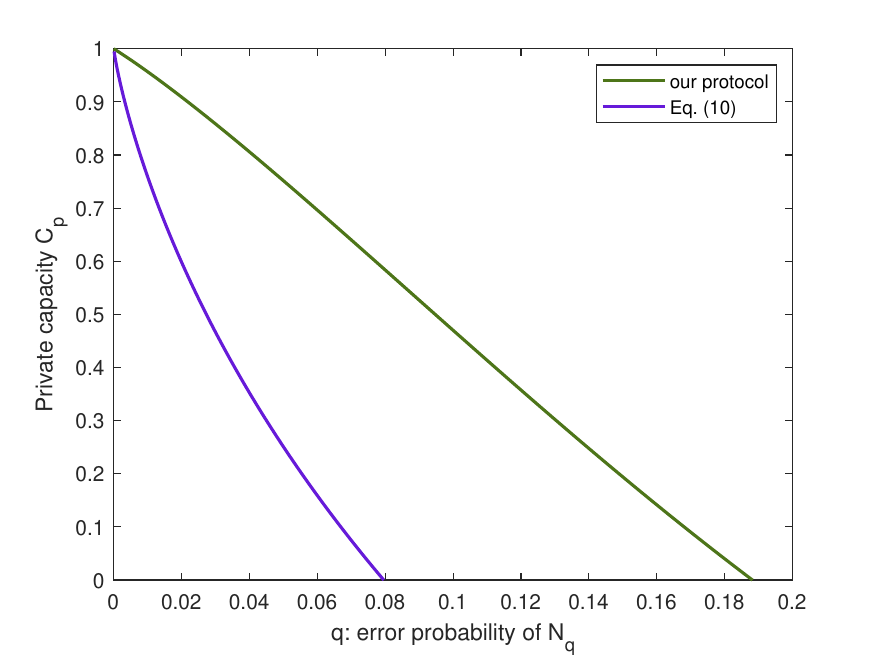}}
	\caption{{\label{Fig3}} Noise reduction in private communication of the quantum $\tt SWITCH$. (a) A diagram of our protocol, in which only two copies of $\sigma_Y \circ \mathcal{N}_q$ are put into quantum $\tt SWITCH$. (b) Private capacity of two channels. The violet curve is an upper bound for private capacity of the composite channel generated by two copies $\mathcal{N}_q$, as given by Eq.~\eqref{BB84}. The green curve corresponds to the coherent information of the output channel of the associated quantum $\tt SWITCH$ in our protocol. }
\end{figure}

Suppose that Alice uses two identical copies of a BB84 channel $\mathcal{N}_q$ to communicate with Bob. If she uses the two copies sequentially, for example, if she sends a message to a repeater who then transmits the message to Bob, then such a strategy is equivalent to using a composite BB84 channel with an error rate of $r = 2q-2q^2$ in a single pass. By Eq.~\eqref{BB84}, it follows that in such a case Alice cannot transmit classical information securely if $q$ is over 8\%. The situation changes however with the aid of the quantum $\tt SWITCH$. As shown in Fig.~\ref{Fig3}, suppose Alice composes both copies of the BB84 channel with a $\sigma_Y$ gate to obtain two copies of an $\mathcal{N}_{1-q}$ channel, and then inputs these two channels into a quantum $\tt SWITCH$. As the coherent information is a lower bound of private capacity, a non-zero coherent information of the output channel of the quantum $\tt SWITCH$ indicates Alice can transmit private classical messages at a non-zero rate. Indeed, the quantum $\tt SWITCH$ can effectively enhance private communication in such a protocol. In particular, when the error probability is at (0.08, 0.188), the conventional scheme with definite causal orders is no longer able to securely transmit classical information, while the scheme making use of the quantum $\tt SWITCH$ still can.

{\em Conclusions.~}
In this Letter, we defined an easily computable quantity $\mathcal{P}_n$ associated with the quantum $\tt SWITCH$ of $n$ channels, and have shown in a large class of examples that the condition $\mathcal{P}_n>0$ is necessary and sufficient for communication enhancement via the quantum $\tt SWITCH$ (outside a set of measure-zero in the associated parameter spaces). In particular, in the case of $n$ copies of a Pauli channel, we were able to prove an explicit formula for the classical capacity of the effective channel associated with the quantum $\tt SWITCH$ of the forward and backward orders, enabling us to effectively compute causal gains associated with the quantum $\tt SWITCH$ in terms of the quantity $\mathcal{P}_n$. We then extended such results to a set of depolarizing qudit channels (of arbitrary dimension) parametrized by an interval, and again showed that $\mathcal{P}_n>0$ is necessary and sufficient for positive causal gains associated with the quantum $\tt SWITCH$ almost surely. Such results then lead us to conjecture that the condition $\mathcal{P}_n>0$ is both necessary and sufficient for communication enhancement via the quantum $\tt SWITCH$ for generic channels. We then concluded by using our results to formulate a communication protocol involving the quantum $\tt SWITCH$ which increases the private capacity of the BB84 channel, which plays a prominent role in quantum cryptography.

{\em Acknowledgements.~} This work is supported by the Fundamental
Research Funds for the Central Universities, the National Natural Science Foundation of China (no. 12371132).

\bibliography{channelcapacity}

\begin{thebibliography}{59}
\bibitem [{foo({\natexlab{a}})}]{footnote1}%
  \BibitemOpen
  \href@noop {} {\bibfield  {journal} {\bibinfo  {journal} {While the quantity
  $\mathcal{P}_n$ depends both on the number of channels $n$ and the
  permutation set $\mathbf{S}$, for ease of notation we have suppressed the
  dependence of $\mathcal{P}_n$ on $\mathbf{S}$}\ }
  ({\natexlab{a}})}\BibitemShut {NoStop}%
\end{thebibliography}

\newpage

\begin{widetext}

\appendix

\section{The condition $\mathcal{P}_n = 0$.}

\subsection{The general case.}

For $n$ copies of quantum channel $\mathcal{C}$ with Kraus operators $\{C_i\}$ and a given set $\textbf{S}$ containing $m$ permutations, the quantum ${\tt SWITCH}$ of the channels $\{\mathcal{C}^{i}\}_{i=1}^{n}$ with respect to the set $\mathbf{S}$ and a control state $\omega$ yields an effective channel $\mathcal{S}^n$ given by
\begin{equation}
	\mathcal{S}^n(\rho) = \sum_{\pi^i,\pi^j\in \mathbf{S}} C_{\pi^i \pi^j}(\rho) \otimes \omega_{i j}|i\rangle\langle j|_C\, ,
\end{equation}
where $C_{\pi^i \pi^j}(\rho) = \sum_{s_1,\dots,s_n} C^{\pi^i(1)}_{s_{\pi^i(1)}}\cdots C^{\pi^i(n)}_{s_{\pi^i(n)}} \rho \bigg( C^{\pi^j(1)}_{s_{\pi^j(1)}}\cdots C^{\pi^j(n)}_{s_{\pi^j(n)}}\bigg)^\dagger$ and $\{C^{k}_{s_k}\}$ are Kraus operators of $k$-th channel $\mathcal{C}$.

As mentioned in main text, for a given control state $\omega = |\omega\rangle\langle\omega|$ with $|\omega\rangle = \sum_{i\in \textbf{S}} |i\rangle /\sqrt{m}$, the quantity $\mathcal{P}_n$ denotes the maximum probability of obtaining the measurement outcome $F_2$ associated with the projective measurement $\{F_1 = |\omega\rangle\langle\omega|, F_2 = I-|\omega\rangle\langle\omega|\}$, i.e.,
\begin{equation}
\mathcal{P}_n = \max_{\rho} \text{Tr}\bigg( (I\otimes F_2)\ \mathcal{S}^n(\rho) \bigg) =  1-\frac{1}{m^2} \min_{\rho} \sum_{i,j} \text{Tr}\bigg(C_{\pi^i \pi^j}(\rho)\bigg)\, .
\end{equation}

In this subsection, we prove the following:
\begin{prop}
    The quantity $\mathcal{P}_n = 0$ if and only if the Kraus operators $\{C_i\}$ are $\textbf{S}$-invariant, that is $C^{\pi^i(1)}_{s_{\pi^i(1)}}\cdots C^{\pi^i(n)}_{s_{\pi^i(n)}} = C^{\pi^j(1)}_{s_{\pi^j(1)}}\cdots C^{\pi^j(n)}_{s_{\pi^j(n)}}$ for all $s_1,\dots,s_n$ and $\pi^i,\pi^j$ are arbitrary two permutations of $\textbf{S}$.
\end{prop}
\begin{proof}
Using Cauchy-Schwarz inequality: $\text{Tr}(AB^\dagger) \leq \sqrt{\text{Tr}(AA^\dagger)} \sqrt{\text{Tr}(BB^\dagger)} \leq \frac{1}{2}\bigg(\text{Tr}(AA^\dagger) + \text{Tr}(BB^\dagger)\bigg)$, the equality holds if and only if $A=B$, and since $\rho$ is positive semi-definite, we have a Hermitian operator $H$ such that $\rho = H^2 = HH^\dagger$, so that $H=\sqrt{\rho}$. Therefore, for arbitrary two permutations $\pi^i,\pi^j \in \textbf{S}$,
\begin{equation}
    \begin{aligned}
       \text{Tr}\bigg(C_{\pi^i \pi^j}(\rho)\bigg) &= \sum_{s_1,\dots,s_n} \text{Tr}\bigg[C^{\pi^i(1)}_{s_{\pi^i(1)}}\cdots C^{\pi^i(n)}_{s_{\pi^i(n)}} \sqrt{\rho} \bigg( C^{\pi^j(1)}_{s_{\pi^j(1)}}\cdots C^{\pi^j(n)}_{s_{\pi^j(n)}} \sqrt{\rho}\bigg)^\dagger\bigg] \\
       & \leq \sum_{s_1,\dots,s_n} \frac{1}{2} \text{Tr}\bigg[ C^{\pi^i(1)}_{s_{\pi^i(1)}}\cdots C^{\pi^i(n)}_{s_{\pi^i(n)}} \sqrt{\rho} \bigg( C^{\pi^i(1)}_{s_{\pi^i(1)}}\cdots C^{\pi^i(n)}_{s_{\pi^i(n)}} \sqrt{\rho}\bigg)^\dagger + C^{\pi^j(1)}_{s_{\pi^j(1)}}\cdots C^{\pi^j(n)}_{s_{\pi^j(n)}} \sqrt{\rho} \bigg( C^{\pi^j(1)}_{s_{\pi^j(1)}}\cdots C^{\pi^j(n)}_{s_{\pi^j(n)}} \sqrt{\rho}\bigg)^\dagger \bigg] \\
       & = \frac{1}{2} \bigg[  \text{Tr}\bigg(C_{\pi^i \pi^i}(\rho)\bigg) +  \text{Tr}\bigg(C_{\pi^j \pi^j}(\rho)\bigg)\bigg] \\
       & = 1\, ,
    \end{aligned}
\end{equation}
where the second equality is for Kraus representation of channels $C_{\pi^i \pi^i}(\rho), C_{\pi^j \pi^j}(\rho)$, the third equality holds because $C_{\pi^i \pi^i}(\rho), C_{\pi^j \pi^j}(\rho)$ are quantum channels, which are trace preserving. Hence, $\mathcal{P}_n = 0$ if and only if $\text{Tr}\bigg(C_{\pi^i \pi^j}(\rho)\bigg) = 1$ holds for all $\rho$ and $\pi^i,\pi^j$.

By the equality condition for Cauchy-Schwarz inequality, we know $\text{Tr}\bigg(C_{\pi^i \pi^j}(\rho)\bigg) = 1$ if and only if $C^{\pi^i(1)}_{s_{\pi^i(1)}}\cdots C^{\pi^i(n)}_{s_{\pi^i(n)}} \sqrt{\rho} = C^{\pi^j(1)}_{s_{\pi^j(1)}}\cdots C^{\pi^j(n)}_{s_{\pi^j(n)}} \sqrt{\rho}$ for all $s_1,\dots,s_n$. Specifically, it is enough to check this equality holds for all pure states $|\psi\rangle\langle\psi|$, that is $C^{\pi^i(1)}_{s_{\pi^i(1)}}\cdots C^{\pi^i(n)}_{s_{\pi^i(n)}} |\psi\rangle\langle\psi| = C^{\pi^j(1)}_{s_{\pi^j(1)}}\cdots C^{\pi^j(n)}_{s_{\pi^j(n)}} |\psi\rangle\langle\psi|$, which is equivalent to $C^{\pi^i(1)}_{s_{\pi^i(1)}}\cdots C^{\pi^i(n)}_{s_{\pi^i(n)}} = C^{\pi^j(1)}_{s_{\pi^j(1)}}\cdots C^{\pi^j(n)}_{s_{\pi^j(n)}}$ for all $s_1,\dots,s_n$.

Finally, as $\text{Tr}\bigg(C_{\pi^i \pi^j}(\rho)\bigg) = 1$ holds for all $\pi^i,\pi^j\in \textbf{S}$, it follows that $\mathcal{P}_n = 0$ if and only if  the Kraus operators $\{C_i\}$ are $\textbf{S}$-invariant, as desired.
\end{proof}
In particular for $n=m=2$, $\textbf{S}$-invariance is equivalent to commutativity.

\subsection{$\mathcal{P}_n = 0$ implies $\delta_f = 0$}
In this subsection, we will show that a channel $\mathcal{C}$, satisfying $\mathcal{P}_n = 0$ for a fixed permutation set $\textbf{S}$, can not increase capacity even with the assistance of quantum $\tt SWITCH$.

As shown above, $\mathcal{P}_n = 0$ if and only if $C^{\pi^i(1)}_{s_{\pi^i(1)}}\cdots C^{\pi^i(n)}_{s_{\pi^i(n)}} = C^{\pi^j(1)}_{s_{\pi^j(1)}}\cdots C^{\pi^j(n)}_{s_{\pi^j(n)}}$ for all $s_1,\dots,s_n$. Hence, for a given subset of permutations $\mathbf{S}$, suppose that $n$ copies of channels $\mathcal{C}$ are put into quantum $\tt SWITCH$, using Eq.~\eqref{switchN} in main text, the effective channel is given by
\begin{equation}
	\mathcal{S}^n(\rho) = \sum_{\pi^k,\pi^l\in \mathbf{S}} C_{\pi^k \pi^l}(\rho) \otimes \omega_{k l}|k\rangle\langle l|_C\, ,
\end{equation}
where
$
	C_{\pi^k \pi^l}(\rho) := \sum_{s_1,\dots,s_n} C_{\pi^k(s_1,\dots,s_n)} \rho C_{\pi^l(s_1,\dots,s_n)}^\dagger
$, and $C_{\pi^k(s_1,\dots,s_n)}:= C^{\pi^k(1)}_{s_{\pi^k(1)}}\cdots C^{\pi^k(n)}_{s_{\pi^k(n)}}$. We can find that $\mathcal{P}_n = 0$ implies that $C_{\pi^k \pi^l}(\rho) = C_{\pi^k \pi^k}(\rho) = C_{\pi^l \pi^l}(\rho) = \mathcal{C}(\rho)$ for every pair of permutations $\pi^k, \pi^l \in \textbf{S}$. 

Therefore, $\mathcal{S}^n(\rho) = \mathcal{C}(\rho) \otimes \omega$, where $\omega$ is a control state and is independent to the input state $\rho$. Hence, the capacity of $\mathcal{S}^n$ is always equal to the capacity of $\mathcal{C}$, that is $\delta_f = 0$.

\subsection{The condition $\mathcal{P}_n = 0$ for Pauli channels.}

In this subsection, we study the condition for $\mathcal{P}_n = 0$ in Pauli channels associated with forward and backward orders. For this, suppose the control qubit is $\omega = |+\rangle\langle+|$, so that by Eq.~\eqref{switchN} in the main text, the effective channel $\mathcal{S}^n$ is given by
\begin{equation}\label{nswitch}
	\begin{aligned}
		\mathcal{S}^n(\rho) & = \frac{1}{4} \bigg(\mathcal{S}_+^n(\rho) \otimes |+\rangle\langle+| + \mathcal{S}_-^n(\rho)\otimes |-\rangle\langle-|\bigg) \, ,
	\end{aligned}
\end{equation}
where the maps $\mathcal{S}_+^n$ and $\mathcal{S}_-^n$ are given by
\begin{equation}
	\mathcal{S}_\pm^n(\rho) = \sum_{\vec{i}\in \{0,...,3\}^n} (C_{\vec{i}}\ \pm\ \widetilde{C}_{\vec{i}})\rho(C_{\vec{i}}\ \pm\  \widetilde{C}_{\vec{i}})^\dagger\, .
\end{equation}
Here, $|\pm\rangle = \big(|0\rangle \pm |1\rangle\big)/2$ is a Fourier bases, $\vec{i} = (i_1,\ldots,i_n)$ is the $n$-tuple of indices, $C_{\vec{i}} = C_{i_1}^1\cdots C_{i_n}^n$ and $\widetilde{C}_{\vec{i}} = C_{i_n}^n\cdots C_{i_1}^1$, where $C_{i_j}^j = \sqrt{p_{i_j}}\sigma_{i_j}$ is the $i_j$-th Kraus operator of $\mathcal{N}$.

Now since $C_{\vec{i}}$ is a product of Pauli matrices and error probabilities, it follows that $C_{\vec{i}} = a_{\vec{i}}\, \sigma_{\vec{i}}$ for some scalar $a_{\vec{i}}$ and some Pauli matrix $\sigma_{\vec{i}}$, from which one may show that $\widetilde{C}_{\vec{i}}= C_{\vec{i}}^\dagger$. Furthermore, as $C_{\vec{i}} - \widetilde{C}_{\vec{i}} = 2i\ \text{Im}(a_{\vec{i}}) \sigma_{\vec{i}}$, measuring the the control system with respect to the Fourier basis will yield $|-\rangle$ with probaility $\mathcal{P}_n(\rho)$ given by
\begin{equation}\label{PnPauli}
	\begin{aligned}
		\mathcal{P}_n(\rho) = \frac{1}{4} \sum_{\vec{i}} \text{Tr}\bigg(\rho \lvert C_{\vec{i}} - \widetilde{C_{\vec{i}}}\rvert^2 \bigg) = \sum_{\vec{i}} \lvert \text{Im}( a_{\vec{i}})\rvert^2 = 1 -  \sum_{\vec{i}} \lvert \text{Re}( a_{\vec{i}})\rvert^2,
	\end{aligned}
\end{equation}
where $\lvert A \rvert ^2 = A^\dagger A$ is Hermitian square of $A$. Since Eq.~\eqref{PnPauli} is independent to the initial state $\rho$, this probability is always equal to $\mathcal{P}_n$. Furthermore, as there always exist some $n$-tuple of indices $\vec{i}$ such that $\text{Re}( a_{\vec{i}}) > 0$, we can know $\mathcal{P}_n < 1$. A more in-depth analysis might give a tighter upper bound for $\mathcal{P}_n$, but this is beyond the scope of this paper, so we ignore it.

Since $\mathcal{P}_n = 0$ if and only if $\text{Im}(a_{\vec{i}}) = 0$ for every $\vec{i}$, we now consider two cases:
\begin{itemize}
	\item[(1)] \underline{$n$ is even}: If there exists two non-commutative Kraus operators of the channel $\mathcal{N}$, for example if $C_1 = \sqrt{p_1} \sigma_1$ and $C_2 = \sqrt{p_2} \sigma_2$, consider the product 
	$$C_{\vec{i}} = C_1^1\cdots C_{1}^{n-1} C_2^n = \sqrt{p_1^{n-1}p_2}\sigma_1^{n-1}\sigma_2 = i\sqrt{p_1^{n-1}p_2} \sigma_3\, , $$ 
	i.e., $\text{Im}(a_{\vec{i}}) =  \sqrt{p_1^{n-1}p_2} >0$, thus $\mathcal{P}_n > 0$. In such a case the Kraus operators of $\mathcal{N}$ are all commutative, so that $\mathcal{N}$ is either a unitary Pauli channel or the Kraus operators of $\mathcal{N}$ are $\sqrt{p_0} I$ and $\sqrt{p_i} \sigma_i$. 
	\item[(2)] \underline{$n$ is odd}: We claim that the number of Kraus operators of $\mathcal{N}$ is less than 3. If there exist at least 3 Kraus operators of $\mathcal{N}$, for example $C_0 = \sqrt{p_0} I$, $C_1 = \sqrt{p_1} \sigma_1$ and $C_2 = \sqrt{p_2} \sigma_2$, then consider the product
	$$C_{\vec{i}} = C_0^1\cdots C_{0}^{n-2}C_1^{n-1} C_2^n = \sqrt{p_0^{n-2}p_1p_2}\sigma_1\sigma_2 = i\sqrt{p_0^{n-2}p_1p_2} \sigma_3\, , $$ 
	hence, $\mathcal{P}_n > 0$.
	
	As such, the number of Kraus operators of $\mathcal{N}$ is not more than 2. If $\mathcal{N}$ is a unitary Pauli channel or its the Kraus operators are all commutative, $\mathcal{P}_n = 0$ is clear. On the other hand, if $\mathcal{N}$ has two non-commutative Kraus operators, for example $C_1 = \sqrt{p_1} \sigma_1$ and $C_2 = \sqrt{p_2} \sigma_2$, as $C_1$ and $C_2$ are anti-commutative, for a fixed index vector $\vec{i}$, we have
	$$C_{\vec{i}} = C_{i_1}^1\cdots C_{i_{n-1}}^{n-1} C_{i_n}^n = (-1)^{s} \sqrt{p_1^{n-r}p_2^r}\sigma_1^{n-r}\sigma_2^r\, , $$ 
	here $s$ depends on the index vector. Since $n$ is odd, we know that $\sigma_1^{n-r}\sigma_2^r$ is $\sigma_1$ if $r$ is even, otherwise it equals $\sigma_2$ when $r$ is odd. Therefore, $|\text{Im}(a_{\vec{i}})|^2 = 0$ for every index vector $\vec{i}$, that is $\mathcal{P}_n = 0$.
\end{itemize}
\begin{figure}
	\includegraphics[width=0.7\textwidth]{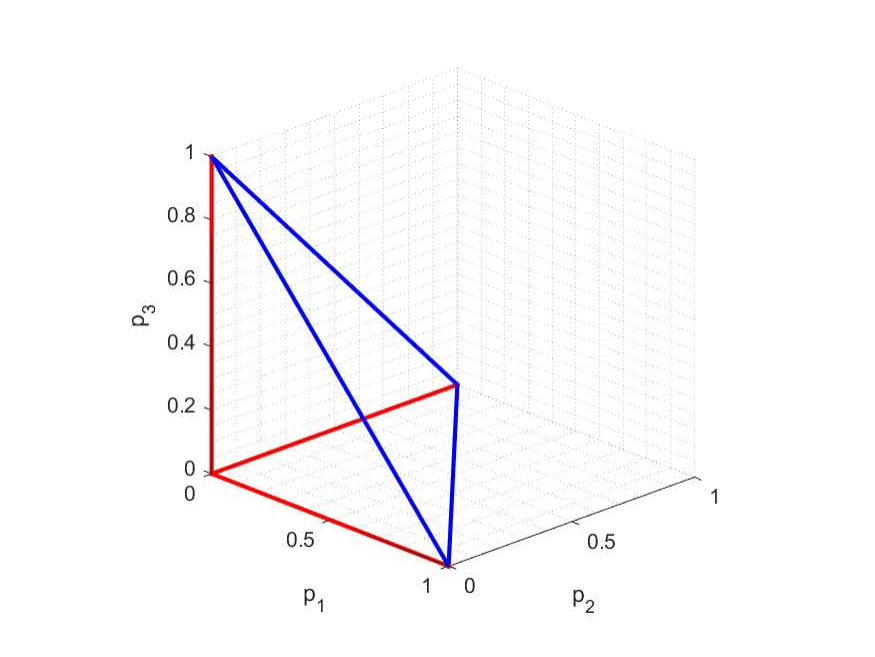}
	\caption{{\label{Fig4}} The set $\mathcal{V}$ consists of all probability vectors of Pauli channels. Three axes represent different noise error probability. Since $\sum_{i=0}^3 p_i = 1$, we can deduce that $\mathcal{V}$ is a tetrahedron. As discussion below, red edges represent the channel with commutative Kraus operators and blue edges indicate Pauli channels contained only two non-commutative Kraus operators. Let $\mathcal{U}$ be the set of Pauli channels satisfying $\mathcal{P}_n=0$. If $n$ is even, these three red edges are the set $\mathcal{U}$; if $n$ is odd, $\mathcal{U}$ is all these six edges. }
\end{figure}

Altogether, we have following:
\begin{prop}\label{PauliPn}
	The quantity $\mathcal{P}_n$ is zero if and only if:
	\begin{itemize}
		\item[(i)] $n$ is even and the Kraus operators of $\mathcal{N}$ are commutative.
		\item[(ii)] $n$ is odd and the number of Kraus operators of $\mathcal{N}$ is not more than 2.
	\end{itemize}
\end{prop}

In both cases, the number of Kraus operators of $\mathcal{N}$ is not more than 2. Such Pauli channels are known either to be degradable or anti-degradable~\cite{degradable,Structure}. If $\mathcal{N}$ is anti-degradable, then $\mathcal{N}^n$ also is, hence its quantum capacity is zero. However, if $\mathcal{N}$ is degradable, $\mathcal{N}^n$ may not be degradable, and its quantum capacity is difficult to determine. On the other hand, although the effective channel $\mathcal{S}^n$ can be regarded as flagged extension of $\mathcal{N}^n$, it is not easy to determine whether $\mathcal{S}^n$ is degradable, which means that evaluating the quantum capacity of $\mathcal{S}^n$ is also difficult.

Let $\mathcal{V} = \{\vec{p}=(p_0,\dots,p_3)| \sum_i p_i =1, p_i\geq 0\}$ be the simplex of all probability vectors of Pauli channels, and let $\mathcal{U}\subset \mathcal{V}$ be the locus satisfying $\mathcal{P}_n = 0$. We show these two sets in Fig.~\ref{Fig4}, the former set $\mathcal{V}$ is the tetrahedron and the latter set $\mathcal{U}$ is all six edges of this tetrahedron. 

We can choose a probability measure $dm$ satisfying $m(\mathcal{V}) = 1$. For example, since this tetrahedron's volume, i.e., its Lebesgue measure $dL$ is 1/6, we can set the probability measure as $dm = 6 dL$, and importantly, $\mathcal{U}$ is a measure-zero set. Hence, the condition $\mathcal{P}_n > 0$ is satisfied almost surely for all $n$.

\section{Calculation of the effective channel}
Although for generic channels it is difficult to give an exact expression of the output channel of quantum $\tt SWITCH$, in the case Pauli channels with forward and backward orders an explicit expression for the effective channel $\mathcal{S}^n$ may be obtained, as we show in this section. Moreover, an important property of $\mathcal{S}^n$ will be provided, which is key to Theorem~\ref{DeltaCX}.

\subsection{An explicit expression for $\mathcal{S}^n(\rho)$}

Let $\mathcal{N}(\rho) = \sum_{i=0}^{3} p_i \sigma_i \rho \sigma_i$ be a Pauli channel. Suppose that there are $n$ copies of $\mathcal{N}$, where $\{C^k_{i_k}\}_{i_k = 0}^3$ is a collection of Kraus operators for the $k$-th channel $\mathcal{N}$, that is $C_{i_k}^k = \sqrt{p_{i_k}}\sigma_{i_k}$, and let $\mathcal{S}^n$ be the effective channel associated with the quantum ${\tt SWITCH}$ of the forward and backward orders, so that
\begin{equation}\label{SS}
	\mathcal{S}^n(\rho) = \frac{1}{4}\bigg(\mathcal{S}_+^n(\rho) \otimes |+\rangle\langle+| + \mathcal{S}_-^n(\rho)\otimes |-\rangle\langle-|\bigg)\, ,
\end{equation}
where
\begin{eqnarray}
	\mathcal{S}_+^n(\rho) = \sum_{i_1\cdots i_n} (C_{i_1}^1\cdots C_{i_n}^n + C_{i_n}^n\cdots C_{i_1}^1)\rho(C_{i_1}^1\cdots C_{i_n}^n + C_{i_n}^n\cdots C_{i_1}^1)^\dagger, \label{S_+} \\
	\mathcal{S}_-^n(\rho) = \sum_{i_1\cdots i_n} (C_{i_1}^1\cdots C_{i_n}^n - C_{i_n}^n\cdots C_{i_1}^1)\rho(C_{i_1}^1\cdots C_{i_n}^n - C_{i_n}^n\cdots C_{i_1}^1)^\dagger\, .
\end{eqnarray}

It is troublesome to evaluate the expression $\mathcal{S}_+^n(\rho)$ and $\mathcal{S}_-^n(\rho)$ directly, since there are $4^n$ terms about the product $C_{i_1}^1\cdots C_{i_n}^n$ and $C_{i_n}^n\cdots C_{i_1}^1$. However, since the product $C_{i_n}^n\cdots C_{i_1}^1$ is just the Hermitian conjugate of the product $C_{i_1}^1\cdots C_{i_n}^n$, and for $i\neq j$ and $i,j \in \{1,2,3\}$, Pauli operators anti-commute, i.e., $\{\sigma_i, \sigma_j\} = 0$, it follows that the product $C_{i_1}^1\cdots C_{i_n}^n = t \sqrt{p_{i_1}\cdots p_{i_n}}\sigma_1^{r_1} \sigma_2^{r_2} \sigma_3^{r_3}$, where $t = \pm 1$, $r_s$ is the number of factors of $\sqrt{p_s}\sigma_s$ appearing in the product $C_{i_1}^1\cdots C_{i_n}^n$ and $r_1 + r_2 + r_3 \leq n$. 

In light of this observation, we classify all the products $C_{i_1}^1\cdots C_{i_n}^n$ according to the number of factors of $\sqrt{p_0}\sigma_0$ they contain, which is equivalent to the number of zeros in the $n$-tuple of indices $(i_1,\cdots,i_n)$, where $i_k \in \{0,1,2,3\}$.

Suppose that there exist $n-k$ such Kraus operators  $\sqrt{p_0}\sigma_0$ in the product $C_{i_1}^1\cdots C_{i_n}^n$, so that the rest of the Kraus operators belong to the set $\{ \sqrt{p_1}\sigma_1,\sqrt{p_2}\sigma_2,\sqrt{p_3}\sigma_3\}$. Moreover, suppose the positions at which these $\sqrt{p_0}\sigma_0$ are located are fixed (there are $\binom{n}{k}$ such cases), and define $T_k$ to be the set containing all $n$-tuples of indices $\vec{i} = (i_1, \cdots, i_n)$ such that the positions where the factors of $\sqrt{p_0}\sigma_0$ are located are the same. For example, if $n=6$ and $k=4$, then we have $\binom{6}{4} = 15$ different possibilities for where $\sqrt{p_0}\sigma_0$ are located, and for example one of these cases is when the Kraus operators of first and fifth channels $\mathcal{N}$ are chosen as $\sqrt{p_0}\sigma_0$, and in such a case we have $T_2 = \{(0,i_2,i_3,i_4,0,i_6) \,|\, i_2,i_3,i_4,i_6 = 1,2,3\}$. 

Now let $C_{\vec{i}} := C_{i_1}^1\cdots C_{i_n}^n$ and $\widetilde{C_{\vec{i}}} := C_{i_n}^n\cdots C_{i_1}^1$. Since $C_{\vec{i}} \widetilde{C_{\vec{i}}} = p_{i_1}p_{i_2}\cdots p_{i_n} I$, we have
\begin{equation*}
\widetilde{C_{\vec{i}}} = C_{\vec{i}}^\dagger = t \sqrt{p_{i_1}\cdots p_{i_n}} \sigma_3^{r_3} \sigma_2^{r_2} \sigma_1^{r_1} = (-1)^{r_1r_2+r_1r_3+r_2r_3}t \sqrt{p_{i_1}\cdots p_{i_n}} \sigma_1^{r_1} \sigma_2^{r_2} \sigma_3^{r_3},
\end{equation*}
from which it follows that 
\begin{equation*}
	C_{\vec{i}} + \widetilde{C_{\vec{i}}}  =  \bigg(1+(-1)^{r_1r_2+r_1r_3+r_2r_3}\bigg) t \sqrt{p_{i_1}\cdots p_{i_n}}  \sigma_1^{r_1} \sigma_2^{r_2} \sigma_3^{r_3}\, .
\end{equation*}

Hence, as discussed above, we can compute the sum of these terms whose $n$-tuple of indices $\vec{i} = (i_1,\cdots,i_n)$ belong to $T_k$, which is:
\begin{equation}\label{CK+}
	C_+^k(\rho) := \sum_{\vec{i}\ \in\ T_k} (C_{\vec{i}} + \widetilde{C_{\vec{i}}})\rho (C_{\vec{i}} + \widetilde{C_{\vec{i}}})^\dagger =  \left\{
	\begin{aligned}
	&	4\bigg[\sum_{\substack{r_1,r_2,r_3\text{are even}\\ r_1+r_2+r_3 = k}} \frac{k!}{r_1!r_2!r_3!} p_1^{r_1}p_2^{r_2}p_3^{r_3}\bigg] \rho =: 4d_0^k \rho, & \text{if $k$ is even} \\
	&	4\sum_{i=1}^3 \bigg[\sum_{\substack{r_i,\text{is odd}\\ \text{others are even}, \\ r_1+r_2+r_3 = k}} \frac{k!}{r_1!r_2!r_3!} p_1^{r_1}p_2^{r_2}p_3^{r_3}\bigg] \sigma_i\rho\sigma_i =: 4\sum_{i=1}^{3} d_i^k \sigma_i\rho\sigma_i, & \text{if $k$ is odd}\, .
	\end{aligned}
	\right.
\end{equation}
It is clear that such $T_k$ depends on the position where $\sqrt{p_0}\sigma_0$ located. However, for a fixed $k$, these products $C_{i_1}^1\cdots C_{i_n}^n$ are only potentially different up to sign, and this difference disappears when the sum is computed, because $$(C_{\vec{i}} + \widetilde{C_{\vec{i}}})\rho (C_{\vec{i}} + \widetilde{C_{\vec{i}}})^\dagger = \bigg(1+(-1)^{r_1r_2+r_1r_3+r_2r_3}\bigg)^2 p_{i_1}\cdots p_{i_n} \sigma_1^{r_1} \sigma_2^{r_2} \sigma_3^{r_3} \rho \bigg(\sigma_1^{r_1} \sigma_2^{r_2} \sigma_3^{r_3}\bigg)^\dagger,$$ which is independent of the sign of $C_{\vec{i}} + \widetilde{C_{\vec{i}}}$. 
Since for a fixed $k$ there are $\binom{n}{k}$ such cases for the positions of $\sqrt{p_0}\sigma_0$, and they all lead a same sum $C_+^k(\rho)$, finally, we have
\begin{equation}\label{SN+}
	\mathcal{S}_+^n(\rho) = \sum_{k=0}^n p_0^{n-k} \binom{n}{k}  C_+^k(\rho)\, .
\end{equation}
Furthermore, we can also calculate a similar expression for $\mathcal{S}_-^n(\rho)$:

\begin{equation}\label{SN-}
	 \mathcal{S}_-^n(\rho) = \sum_{k=2}^n p_0^{n-k} \binom{n}{k}  C_-^k(\rho)\, , 
\end{equation}
where 
\begin{equation}\label{CK-}
	C_-^k(\rho) =  \left\{
	\begin{aligned}
		& 4\sum_{i=1}^3 \bigg[\sum_{\substack{r_i\text{is even}\\ \text{others are odd}\\ r_1+r_2+r_3 = k}} \frac{k!}{r_1!r_2!r_3!} p_1^{r_1}p_2^{r_2}p_3^{r_3}\bigg] \sigma_i\rho\sigma_i =:4\sum_{i=1}^{3} e_i^k \sigma_i\rho\sigma_i, & \text{if $k$ is even}, \\
		& 4\bigg[\sum_{\substack{r_1,r_2,r_3 \text{are odd}\\ r_1+r_2+r_3 = k}} \frac{k!}{r_1!r_2!r_3!} p_1^{r_1}p_2^{r_2}p_3^{r_3}\bigg] \rho =: 4e_0^k \rho, & \text{if $k$ is odd}\, .
	\end{aligned}
	\right.
\end{equation}
Here $d_i^k$ and $e_i^k$ denote the coefficients of $\sigma_i$ for convenience, if a Pauli operator $\sigma_i$ disappears in $C_+^k(\rho)$ or $C_-^k(\rho)$, set $d_i^k = 0$ or $e_i^k = 0$ respectively.

\subsection{Properties of $\mathcal{S}_+^n$ and $\mathcal{S}_-^n$}

In this subsection, we will prove an important property of $\mathcal{S}_+^n$ and $\mathcal{S}_-^n$ which is crucial to our proof of Theorem~\ref{DeltaCX}. By the expressions Eq.~\eqref{SN+} and Eq.~\eqref{SN-}, let $S^n_+(\rho) = 4\sum_{i=0}^{3} s^n_i \sigma_i\rho\sigma_i$ and $S^n_-(\rho) = 4\sum_{i=0}^{3} t^n_i \sigma_i\rho\sigma_i$. We can view  $s_i^n$ and $t_i^n$ as polynomials about $\vec{p}$, where $\vec{p} = (p_0,p_1,p_2,p_3)$ is the probability vector of Pauli channel $\mathcal{N}$. Moreove, these coefficients satisfy the following properties:
\begin{prop} \label{prop1}
	For the coefficient of $\mathcal{S}^n_+$:
	\begin{itemize}
		\item[$\bullet$] If $n$ is even, then 
		\begin{itemize}
			\item[(1)] 	$	s_0^n - s_i^n \geq 0$ for $i\in\{1,2,3\},$ and
			equality holds if and only if $p_0 = p_i = 1/2$, and other probability values are zero.
			\item[(2)] for distinct $j,k\in\{1,2,3\}$, $p_0>0$ implies $s_j^n - s_k^n = (p_j-p_k)f_1^{jk}(\vec{p})$, where $f_1^{jk}(\vec{p})$ depends on $j,k$ and is always positive. If $p_0 = 0$, then $s_j^n = s_k^n = 0$.
		\end{itemize}
		\item[$\bullet$] If $n$ is odd, then
		\begin{itemize}
			\item[(3)] $s_0^n - s_i^n = (p_0-p_i)f_2^{i}(\vec{p})$ for $i\in\{1,2,3\}$ and $f_2^{i}(\vec{p})$ depends on $i$ and is always positive.
			\item[(4)]  $s_j^n - s_k^n = (p_j-p_k)f_3^{jk}(\vec{p})$ for $j,k\in\{1,2,3\}$ and $f_3^{jk}(\vec{p})$ depends on $j,k$ and is always positive.
		\end{itemize}
	\end{itemize}
\end{prop}
\begin{proof}
	We first consider the case when $n$ is even. Without loss of generality, we only prove $s_0^n - s_1^n \geq 0$ and $s_3^n - s_2^n = (p_3-p_2)f(\vec{p})$ for $p_0>0$, since if $p_0 = 0$, by Eq.~\eqref{SN+} and \eqref{CK+}, one can see $s_3^n = s_2^n = 0$.

 Since $s_0^n = \sum_{k=0}^{n/2} d_0^{2k}$ and $s_1^n = \sum_{k=0}^{\frac{n}{2}-1} d_0^{2k+1}$ can be viewed as polynomials about $p_2,p_3$, and exponents of $p_2,p_3$ in $d_0^{2k},d_1^{2k+1}$ are all even, so we can consider the coefficient of $p_2^{2i}p_3^{2m-2i}$ in $s_0^n-s_1^n$, which is 
	\begin{equation}
		\begin{aligned}
		& \sum_{l=2m}^n (-1)^l\binom{n}{l}
		\frac{l!}{(l-2m)!(2i)!(2m-2i)!}p_0^{n-l}p_1^{l-2m}\\
		& =\binom{n}{2m}\binom{2m}{2i} \sum_{l=0}^{n-2m} (-1)^l\binom{n-2m}{n-l-2m}p_0^{n-l-2m}p_1^{l}\\
		& = \binom{n}{2m}\binom{2m}{2i} (p_0-p_1)^{n-2m}\geq 0\, ,
		\end{aligned}
	\end{equation}
	where $\binom{n}{k} = \frac{n!}{(n-k)!k!}$ is the binomial coefficient. We then have 
	\begin{equation}
		s_0^n-s_1^n = \sum_{m=0}^{n/2}\binom{n}{2m} (p_0-p_1)^{n-2m}\bigg[\sum_{i=0}^{m}\binom{2m}{2i} p_2^{2i}p_3^{2m-2i}\bigg] \geq 0\, ,
	\end{equation} 
	with equality if and only if $p_0 = p_1$ and $p_2 = p_3 =0$.
	
	On another hand, $s_3^n-s_2^n$ can be viewed as a polynomial about $p_1$ whose exponents are even. Moreover the coefficient of $p_1^{2i}$ is
	\begin{equation}
		\begin{aligned}
			& \sum_{k=i}^{\frac{n-2}{2}} \binom{n}{2k+1}p_0^{n-2k-1} \sum_{j=0}^{k-i} \frac{(2k+1)!}{(2i)!(2j)!(2k+1-2i-2j)!} (p_2^{2j}p_3^{2k+1-2j-2i} - p_3^{2j}p_2^{2k+1-2j-2i}) \\
			& = (p_3-p_2) \sum_{k=i}^{\frac{n-2}{2}} \binom{n}{2k+1}p_0^{n-2k-1} \sum_{j=0}^{k-i} \frac{(2k+1)!}{(2i)!(2j)!(2k+1-2i-2j)!} p_2^{2j}p_3^{2k-2j-2i} \\
			& =:(p_3-p_2) g_{i}(p_0,p_2,p_3)\, ,
		\end{aligned}
	\end{equation} 
	thus $g_{i}$ equals to zero if and only if $p_0 = 0$. Hence, 
	\begin{equation}
		s^n_3-s^n_2 = (p_3-p_2) \sum_{i=0}^{n/2-1} p_1^{2i}g_i \, .
	\end{equation}
	If $p_2 = p_3$, then $s_3^n-s_2^n = 0$ and we can set $f_1^{23}(\vec{p}) \equiv 1$. If $p_2 \neq p_3$, as $g_i$ defined above is always positive, then we can set $f_1^{23}(\vec{p}) = \sum_{i=0}^{n/2-1} p_1^{2i}g_i$, which is also positive.
	
	Now we consider the case when $n$ is odd. Similarly, we only consider $s^n_0 - s^n_1$ and $s^n_3 -s^n_2$. For $s^n_0 -s^n_1$, we also consider the coefficient of $p_2^{2i}p_3^{2m-2i}$, which is equal to $\binom{n}{2m}\binom{2m}{2i} (p_0-p_1)^{n-2m}$, so
	\begin{equation}
		s^n_0 - s^n_1 = \sum_{m=0}^{\frac{n-1}{2}}\binom{n}{2m} (p_0-p_1)^{n-2m}\bigg[\sum_{i=0}^{m}\binom{2m}{2i} p_2^{2i}p_3^{2m-2i}\bigg]\, .
	\end{equation}
	As $n$ is odd, the sign of $s^n_0 - s^n_1$ is same as the sign of $p_0-p_1$, and $s^n_0 - s^n_1 =0$ if and only if $p_0 = p_1$.
	
	For $s^n_3 -s^n_2$, the coefficient of $p_1^{2i}$ is $$(p_3-p_2) \sum_{k=i}^{\frac{n-1}{2}} \binom{n}{2k+1}p_0^{n-2k-1} \sum_{j=0}^{k-i} \frac{(2k+1)!}{(2i)!(2j)!(2k+1-2i-2j)!} p_2^{2j}p_3^{2k-2j-2i},$$
	so that the conclusion follows \emph{mutatis mutandis} as in the case when $n$ is even, thus concluding the proof.
\end{proof}

We now prove a similar result for $\mathcal{S}^{n}_{-}$:
\begin{prop}\label{prop2}
	 For the coefficient of $\mathcal{S}^n_-$:
	\begin{itemize}
		\item[$\bullet$] If $n$ is even, then 
		\begin{itemize}
			\item[(1)] 	$	t_0^n - t_i^n \leq 0$ for $i\in\{1,2,3\}$, and
			equality holds if and only if $p_0 = p_i = 1/2$, and other probability values are zero.
			\item[(2)] for distinct $i,j,k\in\{1,2,3\}$, $p_i >0$ implies $t_j^n - t_k^n = -(p_j-p_k)\widetilde{f}_1^{jk}(\vec{p})$, where $\widetilde{f}_1^{jk}(\vec{p})$ depends on $j,k$ and is always positive; if $p_i = 0$, then $t_j^n = t_k^n = 0$.
		\end{itemize}
		\item[$\bullet$] If $n$ is odd, then
		\begin{itemize}
			\item[(3)] $t_0^n - t_i^n = -(p_0-p_i)\widetilde{f}_2^{i}(\vec{p})$ for $i\in\{1,2,3\}$ and $\widetilde{f}_2^{i}(\vec{p})$ depends on $i$ and is always positive.
			\item[(4)] for distinct $i,j,k\in\{1,2,3\}$, $p_0p_i>0$ implies $t_j^n - t_k^n = -(p_j-p_k)\widetilde{f}_3^{jk}(\vec{p})$, where $\widetilde{f}_3^{jk}(\vec{p})$ depends on $j,k$ and is always positive; if either $p_i=0$ or $p_0 = 0$, then $t_j^n = t_k^n = 0$.
		\end{itemize}
	\end{itemize}
\end{prop}
The proof is similar to Proposition~\ref{prop1} so we provide a sketch as follows:
\begin{proof}
Also, we just prove the cases $t_0^n-t_1^n$ and $t_3^n-t_2^n$.

If $n$ is even, one can see $$t_0^n-t_1^n = -\sum_{m=1}^{n/2}\binom{n}{2m} (p_0-p_1)^{n-2m}\bigg[\sum_{i=0}^{m-1}\binom{2m}{2i+1} p_2^{2i+1}p_3^{2m-2i-1}\bigg] \leq 0\, , $$
and if $p_1 = 0$, it is clear to see $t_3^n=t_2^n=0$ by Eq.~\eqref{SN-} and \eqref{CK-}. If $p_1>0$, we have
$$t^n_3-t^n_2 = (p_2-p_3) \sum_{i=0}^{n/2-1}p_1^{2i+1}\sum_{k=i+1}^{\frac{n}{2}} \binom{n}{2k}p_0^{n-2k} \sum_{j=0}^{k-i-1} \frac{(2k)!}{(2i+1)!(2j)!(2k-1-2i-2j)!} p_2^{2j}p_3^{2k-2j-2i-2}. $$
The proof is then completed by the same argument as Proposition~\ref{prop1}.

If $n$ is odd, we have $$t_0^n-t_1^n = -\sum_{m=1}^{\frac{n-1}{2}}\binom{n}{2m} (p_0-p_1)^{n-2m}\bigg[\sum_{i=0}^{m-1}\binom{2m}{2i+1} p_2^{2i+1}p_3^{2m-2i-1}\bigg]\, , $$ and 
$$t^n_3-t^n_2 = (p_2-p_3) \sum_{i=0}^{\frac{n-3}{2}}p_1^{2i+1}\sum_{k=i+1}^{\frac{n-1}{2}} \binom{n}{2k}p_0^{n-2k} \sum_{j=0}^{k-i-1} \frac{(2k)!}{(2i+1)!(2j)!(2k-1-2i-2j)!} p_2^{2j}p_3^{2k-2j-2i-2}\, , $$
from which we obtain the result.
\end{proof}

\section{Classical capacity of $\mathcal{S}^n$ produced by Pauli channels with forward and backward orders }

Let $\mathcal{N}(\rho) = \sum_{i=0}^{3} p_i \sigma_i \rho \sigma_i$ be a Pauli channel, and $\mathcal{N}^n(\rho) = \sum_{i=0}^{3} q_i \sigma_i \rho \sigma_i$ be the $n$-fold composition $\mathcal{N}\circ\cdots\circ\mathcal{N}$. Suppose that the $n$ copies of $\mathcal{N}$ are put into the quantum $\tt SWITCH$ in a superposition of forward ($\mathcal{N}_n \circ\cdots\circ\mathcal{N}_1$) and backward ($\mathcal{N}_1 \circ\cdots\circ\mathcal{N}_n$) orders, where $\mathcal{N}_i = \mathcal{N}$ for all $i$, and the control state is $\omega = |+\rangle\langle+|$. So by Eq.~\eqref{switchN} in the main text, the effective channel is 
\begin{equation}
	\begin{aligned}
		\mathcal{S}^n(\rho) & = \frac{1}{4} \bigg(\mathcal{S}_+^n(\rho) \otimes |+\rangle\langle+| + \mathcal{S}_-^n(\rho)\otimes |-\rangle\langle-|\bigg) \, .
	\end{aligned}
\end{equation}
As computed in Eq.~(\ref{SN+}), (\ref{SN-}), the two maps $\mathcal{S}_+^n(\rho)$ and $\mathcal{S}_-^n(\rho)$ are Pauli maps in the sense that $S^n_+(\rho) = 4\sum_{i=0}^{3} s^n_i \sigma_i\rho\sigma_i$ and $S^n_-(\rho) = 4\sum_{i=0}^{3} t^n_i \sigma_i\rho\sigma_i$, where $s^n_i$ and $t^n_i$ are non-negative numbers for all $i$. Moreover, by Eq.~\eqref{PnPauli} in Appendix A we have 
\begin{equation}
    \mathcal{P}_n(\rho) = \frac{1}{4}\text{Tr}(\mathcal{S}^n_-(\rho)) = \sum_{i=0}^3 t_i^n\, ,
\end{equation}
which is independent of the input state $\rho$, and as such, is equal to the quantity $\mathcal{P}_n$ as defined in Eq.~\eqref{generalPn} of main text. Therefore, we can rewrite the effective channel $\mathcal{S}^n$ as
\begin{equation}\label{nSwitch}
	\mathcal{S}^n(\rho) = \big(1-\mathcal{P}_n\big) \Phi_+(\rho)  \otimes |+\rangle\langle+| \ + \  \mathcal{P}_n\Phi_-(\rho) \otimes |-\rangle\langle-|\, .
\end{equation}
If $\mathcal{P}_n = 0$, we can see $\mathcal{S}^n(\rho) = \mathcal{N}^n(\rho) \otimes |+\rangle\langle+|$. Since the control state $|+\rangle\langle+|$ is independent of initial state $\rho$, the classical capacity of  $\mathcal{S}^n$ is always equal to that of $\mathcal{N}^n$. Hence, in this section, unless stated otherwise, we always assume that $\mathcal{P}_n > 0$ ($\mathcal{P}_n < 1$ always holds, see discussion following Eq.~\eqref{PnPauli}).

\subsection{An upper bound for classical capacity of $\mathcal{S}^n$}
In this subsection, we give an upper bound for classical capacity of $\mathcal{S}^n$ following the proof of section VIII in Supplemental Material of~\cite{QSWofCDC}. According to the Holevo-Schumacher-Westmoreland theorem~\cite{HSW1,HSW2}, the classical capacity of a generic quantum channel $\mathcal{M}$ is
$$
C(\mathcal{M})=\liminf _{m \rightarrow \infty} \frac{\chi\left(\mathcal{M}^{\otimes m}\right)}{m},
$$
where $\chi(\mathcal{M})$ is the Holevo information of $\mathcal{M}$, which is defined as 
\[
\chi(\mathcal{M}):=\sup _{\{p_x,\rho_x\}} S\left(\mathcal{M}(\sum_x p_x \rho_x)\right)-\sum_x p_x S\left(\mathcal{M}(\rho_x)\right)\, ,
\]
with the maximum being over all ensembles $\{p_x,\rho_x\}$, where $\rho_x$ is a quantum state and $p_x$ is a probability.

We now consider the effective channel $\mathcal{S}^n$, which has the form 
$$\mathcal{S}^n(\rho)=\big(1-\mathcal{P}_n\big) \Phi_+(\rho)  \otimes |+\rangle\langle+| \ + \  \mathcal{P}_n\Phi_-(\rho) \otimes |-\rangle\langle-|\, , $$ 
where $\Phi_+(\rho) = \sum_{i=0}^{3} s_i \sigma_i \rho \sigma_i$ and $\Phi_-(\rho) = \sum_{i=0}^n t_i\sigma_i \rho \sigma_i$ are two Pauli channels.
It then follows that $m$-fold tensor product $(\mathcal{S}^n)^{\otimes m}$ has the form
$$
(\mathcal{S}^n)^{\otimes m}=\sum_j p_j^{(m)} \mathcal{E}_j^{(m)}\otimes \rho_j^{(m)},
$$
where for each $j$, $\mathcal{E}_j^{(m)}$ is the tensor product of $k$ copies of channel $\Phi_+$ and $(m-k)$ copies of the channel $\Phi_-$ for some $k \in\{0, \ldots, m\})$, and $\rho_j^{(m)}$ is the tensor product of $k$ copies of the state $|+\rangle\langle+|$ and $(m-k)$ copies of the state $|-\rangle\langle-|$ for some $k \in\{0, \ldots, m\})$. For example, for $m=2$ we have
\begin{eqnarray*}
(\mathcal{S}^n)^{\otimes 3}&=(1-\mathcal{P}_n)^2 \Phi_+\otimes \Phi_+\otimes |++\rangle \langle ++ |+(1-\mathcal{P}_n)\mathcal{P}_n \Phi_+\otimes \Phi_- \otimes |+-\rangle \langle +- | \\
&+(1-\mathcal{P}_n)\mathcal{P}_n \Phi_-\otimes \Phi_+ \otimes |-+\rangle \langle -+ |+\mathcal{P}_n^2 \Phi_-\otimes \Phi_- \otimes |--\rangle \langle -- |\, .
\end{eqnarray*}

By convexity in the signal state of the Holevo information, one has the inequality
$$
\begin{aligned}
	\chi\left((\mathcal{S}^n)^{\otimes m}\right) & \leq \sum_j p_j^{(m)} \chi\left(\mathcal{E}_j^{(m)}\otimes \rho_j^{(m)}\right) \\
	& =\sum_j p_j^{(m)} \chi\left(\mathcal{E}_j^{(m)}\right)\, ,
\end{aligned}
$$
where the second equality follows from the fact that the state $\rho_j^{(m)}$ is independent of the input of the channel.

Now since Holevo information is additive over channels consisting of a single qubit~\cite{Additivity}, and $\Phi_+$ and $\Phi_-$ are unit qubit channels, we have
\begin{equation}
	\chi\left(\mathcal{E}_j^{(m)}\right) = k\ \chi(\Phi_+) + (m-k)\ \chi(\Phi_-)\, ,
\end{equation}
where $k$ is the number of copies of $1-\mathcal{P}_n$ appearing in $p_j^{(m)}$. We now let $\mu$ and $\nu$ be the maximum of the set of absolute values of the eigenvalues of $\Phi_+$ and $\Phi_-$ respectively, so that the Holevo information of $\Phi_+$ and $\Phi_-$ are 
\begin{equation}
	\chi(\Phi_+) = 1 - h(\mu) \qquad \text{and} \qquad \chi(\Phi_-) = 1 - h(\nu)\, ,
\end{equation}
where $h(x) = -\frac{1+x}{2}\log \frac{1+x}{2} - \frac{1-x}{2}\log \frac{1-x}{2}$. We then obtain the bound
\begin{equation}
	\begin{aligned}
		\chi\left((\mathcal{S}^n)^{\otimes m}\right) & \leq \sum_{k=0}^m (1-\mathcal{P}_n)^k \mathcal{P}_n^{m-k}\binom{m}{k}\left[ k\ \chi(\Phi_+) + (m-k)\ \chi(\Phi_-)\right] \\
		& =m\bigg[(1-\mathcal{P}_n)\chi(\Phi_+) + \mathcal{P}_n\chi(\Phi_-) \bigg]\, ,
	\end{aligned}
\end{equation}
and therefore
\begin{equation}\label{upp}
	\begin{aligned}
		C(\mathcal{S}^n) &\leq (1-\mathcal{P}_n)\  \chi(\Phi_+) + \mathcal{P}_n\ \chi(\Phi_-) \\
		& = (1-\mathcal{P}_n)\bigg[1-h(\mu)\bigg] +
		\mathcal{P}_n\bigg[1-h(\nu)\bigg] \\
		& = 1 - (1-\mathcal{P}_n)h(\mu)-\mathcal{P}_n h(\nu)\, .
	\end{aligned}
\end{equation}

\subsection{Holevo information of $\mathcal{S}^n$}
Since Holevo information is always less than classical capacity, we can compute the Holevo information of $\mathcal{S}^n$ to give a lower bound for its classical capacity. The proof below follows the Supplemental Material of~\cite{CSwitchPRL} and~\cite{Wilde}. Since it is sufficient to consider optimizing the Holevo information of $\mathcal{S}^n$ over a classical-quantum state with conditional states that are pure, and the Holevo information is $\max_\rho I(X;B)_{\sigma}$ where $\sigma_{XB}$ is the output state, we have
\begin{equation}
	\sigma_{XBC} = \sum_{x} p_x |x\rangle\langle x|_X \otimes \mathcal{S}^n(\psi_A^x) ,
\end{equation}
here the lower subscripts denote explicitly the Hilbert spaces of the state. We consider an extend input state of the form
\begin{equation}
	\omega_{XIJAC} = \frac{1}{4} \sum_{x,i,j} p_x |x\rangle\langle x|_X \otimes |i\rangle\langle i|_I \otimes |j\rangle\langle j|_J\otimes X^iZ^j\psi_A^x Z^j X^i ,
\end{equation}
where system $I$ and $J$ are two new classical registers and $\psi^x_A$ are some pure states of the target system. Moreover, the probability distribution for the registers $I$ and $J$ is assumed to be uniform, and $X$ and $Z$ are the usual Pauli operators.

The mutual information $I(X;BC)_\sigma$ can then be bounded as
\begin{equation}\label{MutualInf}
	\begin{aligned}
		I(X;BC)_\sigma & = H(BC)_\sigma - H(BC|X)_\sigma \\
		& \leq H(BC)_{\mathcal{S}^n(\omega)} - H(BC|X)_\sigma \\
		& = 1 + H(\mathcal{P}_n) - H(BC|X)_\sigma\, ,
	\end{aligned}
\end{equation}
where the first inequality follows from concavity of the von Neumann entropy, and the final equality follows from the fact that
\begin{equation}
	\text{Tr}_{XIJ}[(I \otimes\mathcal{S}^n)(\omega_{XIJAC})] = \frac{I}{2} \otimes \bigg((1-\mathcal{P}_n)|+\rangle\langle +| + \mathcal{P}_n|-\rangle\langle -| \bigg),
\end{equation}
where $ H(\mathcal{P}_n)$ is the binary entropy of $\mathcal{P}_n$.
To further bound the mutual information from above, we analyse the conditional entropy $H(BC|X)_\sigma$, for which we have
\begin{equation}
	\begin{aligned}
		H(BC|X)_\sigma & =  \sum_x p_x H(BC)_{\mathcal{S}^n(\psi_A^x)}\\
		& =  \frac{1}{4}\sum_{x,i,j} p_x H(BC)_{[(X^iZ^j)\otimes I] \mathcal{S}^n(\psi_A^x ) [(Z^j X^i)\otimes I]} \\
		& = \frac{1}{4} \sum_{x,i,j} p_x H(BC)_{\mathcal{S}^n(X^iZ^j \psi_A^x Z^jX^i )} \\
		& = H(BC|XIJ)_{\mathcal{S}^n(\omega_{XIJAC})}\, ,
	\end{aligned}
\end{equation}
where the second equality follows from the fact that the von Neumann entropy is invariant under isometric transformations, and the third equality follows from the fact that $X^i$ and $Z^j$ act solely on the system state while leaving the control state fixed. Hence, we can bound Eq.~\eqref{MutualInf} as 
\begin{equation}
	\begin{aligned}
		I(X;BC)_\sigma &\leq 1 + H(\mathcal{P}_n) - H(BC|XIJ)_{\mathcal{S}^n(\omega)} \\
		& = 1 + H(\mathcal{P}_n) -  \sum_x p_x H(BC)_{\mathcal{S}^n(\psi_A^x)} \\
		& \leq 1 + H(\mathcal{P}_n) - \min_x H(BC)_{\mathcal{S}^n(\psi_A^x)} \\
		& \leq 1 + H(\mathcal{P}_n) - H^{\text{min}}(\mathcal{S}^n)\, ,
	\end{aligned}
\end{equation}
where the first inequality follows from the fact that the expectation value can never be smaller than the minimum value, and in the last equality, $H^{\text{min}}(\mathcal{S}^n) := \min_\rho H(\mathcal{S}^n(\rho))$.

In what follows, we compute $H^{\text{min}}(\mathcal{S}^n)$. From Eq.~\eqref{nSwitch}, we can denote the right hand side as the matrix
\begin{equation}
	\mathcal{S}^n(\rho) = \frac{1}{2}
	\begin{pmatrix}
		A+B & A-B\\
		A-B & A+B
	\end{pmatrix}\, ,
\end{equation}
where $A = (1-\mathcal{P}_n) \Phi_+(\rho)$ and $B = \mathcal{P}_n\Phi_-(\rho)$. The eigenvalues of such a matrix are the union of eigenvalues of $A$ and $B$. Therefore, finding the eigenvalues of $A$ and $B$ is sufficient. Suppose 
\begin{equation*}
	\rho = \frac{1}{2} 
	\begin{pmatrix}
		1+z & x-iy\\
		x+iy & 1-z
	\end{pmatrix}
	= \frac{1}{2}(I+xX+yY+zZ)\, ,
\end{equation*}
where $r = (x,y,z)$ is a real vector. We then have 
\begin{equation}
	\begin{aligned}
		A & = (1-\mathcal{P}_n) \Phi_+(\rho) = (1-\mathcal{P}_n)\sum_{i=0}^3 s_i \sigma_i \rho \sigma_i\\
		& = \frac{1}{2}(1-\mathcal{P}_n)
		\begin{pmatrix}
			\mu_0 + \mu_3 z & \mu_1 x - \mu_2 iy\\
			\mu_1 x + \mu_2 iy & \mu_0 - \mu_3 z
		\end{pmatrix}\, ,
	\end{aligned}
\end{equation}
where $\mu_0 = \sum_{i=0}^3 s_i = 1$ and $\mu_i = s_0 + 2s_i - \sum_{j=1}^3 s_j, i\in\{1,2,3\}$, are the eigenvalues of the Pauli channel $\Phi_+$. It then follows that the eigenvalues of $A$ and $B$ may be written as  
\begin{equation*}
	\lambda_{\pm}^A = \frac{1}{2}(1-\mathcal{P}_n)\bigg(1 \pm \sqrt{\mu_1^2x^2 + \mu_2^2y^2 + \mu_3^2 z^2}\bigg)
\end{equation*}
and 
\begin{equation*}
	\lambda_{\pm}^B = \frac{1}{2}\mathcal{P}_n\bigg(1 \pm \sqrt{\nu_1^2x^2 + \nu_2^2 y^2 + \nu_3^2 z^2}\bigg)\, ,
\end{equation*}
where $\nu_i$ are eigenvalues of Pauli channel $\Phi_-$, i.e. $\nu_0 = 1, \nu_i = t_0 + 2t_i -\sum_{j=1}^3 t_j$ for $i=1,2,3$. 

Now, the minimal entropy of $\mathcal{S}^n$ is 
\begin{equation}
	\begin{aligned}
		H^{\text{min}}(\mathcal{S}^n) & = \min_\rho H(\mathcal{S}^n(\rho)) \\
		& = \min_\rho -\{\lambda_{\pm}^A \log\lambda_{\pm}^A + \lambda_{\pm}^B\log\lambda_{\pm}^B\} \\
		& = \min_\rho \bigg( H(A)+H(B)\bigg) \\
		& = H(\mathcal{P}_n) + \min_\rho \left[ (1-\mathcal{P}_n) h\big(\sqrt{K}\big) + \mathcal{P}_n h\big(\sqrt{K'}\big) \right]\, ,
	\end{aligned}
\end{equation}
where $K = \mu_1^2x^2 + \mu_2^2y^2 + \mu_3^2 z^2$ and $K' = \nu_1^2 x^2 + \nu_2^2 y^2 + \nu_3^2 z^2$, and $$h(x) = -\frac{1+x}{2}\log \frac{1+x}{2} - \frac{1-x}{2}\log \frac{1-x}{2}.$$ 
The minimal output entropy always decreases with respect to $x,y,z$, so we have 
\begin{equation}
	H^{\text{min}}(\mathcal{S}^n) =  H(\mathcal{P}_n) + \min_{i=1,2,3} \left[ (1-\mathcal{P}_n) h\big(|\mu_i|\big) + \mathcal{P}_n h\big(|\nu_i|\big) \right]\, .
\end{equation}
Since $\nu_i$ and $\mu_i$ are fixed, $x^2+y^2+z^2 \leq 1$, and the entropy is a concave function, the minimum value is attained when $x^2+y^2+z^2 = 1$, i.e., $\rho$ is a pure state. Hence, we obtain the following the upper bound for the mutual information: 
\begin{equation}
	I(X;BC) \leq 1 + H(\mathcal{P}_n) - H^{min}(\mathcal{S}^n)\, .
\end{equation}

Indeed, when $\rho_+ = \frac{1}{2}(I + \vec{r}\cdot \vec{\sigma}), \ |\vec{r}| =1$ minimizes the minimal entropy, the pure state $\rho_- = \frac{1}{2}(I - \vec{r}\cdot \vec{\sigma})$ also achieves the minimum. 
Therefore, we can choose the state ensemble as $\rho_+$ and $\rho_-$ with equal probability to achieve this bound, i.e. $\sigma_{XBC} = \frac{1}{2} \sum_\pm |\pm\rangle\langle \pm|\otimes \mathcal{S}^n(\rho_\pm)$. As such, it follows that the Holevo information of $\mathcal{S}^n$ is 
\begin{equation}\label{HoveloSn}
	\chi(\mathcal{S}^n) = 1 -\min_{i=1,2,3} \bigg[ (1-\mathcal{P}_n) h\big(|\mu_i|\big) + \mathcal{P}_n h\big(|\nu_i|\big) \bigg]\, .
\end{equation}

\subsection{The classical capacity of $\mathcal{S}^n$ equals its Holevo information}
In this subsection, we will show that the classical capacity of $\mathcal{S}^n$ equals its Holevo information. For this, we only need to prove that the upper bound for the classical capacity of $\mathcal{S}^n$ as given by Eq.~\eqref{upp} is equal to its Holevo information, which is given by formula ~\eqref{HoveloSn}. Moreover, these two bounds are equal if and only if the minimal output entropy of $\Phi_+$ and $\Phi_-$ are attained at the same input state, i.e., there exists a fixed $i$ such that $\mu = |\mu_i|$ and $\nu = |\nu_i|$. Now we use the properties of $\Phi_+,\Phi_-$ to prove this simple fact, which we will need for the proof of Theorem~\ref{thm1}.

\begin{thm}
	The minimal output entropy of $\Phi_+$ and $\Phi_-$ are attained in the same input state, hence $C(\mathcal{S}^n) = (1-\mathcal{P}_n) C(\Phi_+) + \mathcal{P}_n C(\Phi_-)$.
\end{thm}
\begin{proof}
If $\mathcal{P}_n = 0$, we have $\mathcal{S}^n(\rho) = \Phi_+(\rho) \otimes |+\rangle\langle+|$ and the conclusion is obvious. 
 
Now suppose that $\mathcal{P}_n > 0$. Let $\Phi_+(\rho) = \sum_{i=0}^{3} s_i \sigma_i \rho \sigma_i$ and $\Phi_-(\rho) = \sum_{i=0}^n t_i\sigma_i \rho \sigma_i$. Recall that eigenvalues of $\Phi_+$ are $\mu_i = s_0 + 2s_i - \sum_{j=1}^3 s_j$ and that the eigenvalues of $\Phi_-$ are $\nu_i = t_0 + 2t_i - \sum_{j=1}^3 t_j$ for $i=1,2,3$. Denote $\mu = \max_i |\mu_i|$ and $\nu = \max_i |\nu_i|$.
	
Fix $n$. Without loss of generality, we show that if $\mu = |\mu_3| \geq \max \{ |\mu_1|, |\mu_2|\}$ then $\nu = |\nu_3| \geq \max \{|\nu_1|,|\nu_2|\}$. As $\Phi_+ = \frac{1}{4(1-\mathcal{P}_n)} \mathcal{S}^n_+$ and $\Phi_- = \frac{1}{4\mathcal{P}_n} \mathcal{S}^n_-$, where $\mathcal{S}^n_+(\rho) = 4\sum_i s^n_i \sigma_i\rho\sigma_i$ and $\mathcal{S}^n_-(\rho) = 4\sum_i t^n_i \sigma_i\rho\sigma_i$, we have
	\begin{equation}\label{mu3}
		\begin{aligned}
		\mu_3^2 \geq \mu_2^2 \Longleftrightarrow  (s^n_0 - s^n_1)(s^n_3 - s^n_2) \geq 0 \qquad \text{and} \qquad
		\mu_3^2 \geq \mu_1^2 \Longleftrightarrow  (s^n_0 - s^n_2)(s^n_3 - s^n_1) \geq 0\, .
		\end{aligned}
	\end{equation} 
We then need to prove that $\nu = |\nu_3| \geq \max \{|\nu_1|,|\nu_2|\}$, i.e.,
 	\begin{equation}
 	\begin{aligned}
 		\nu_3^2 \geq \nu_2^2 \Longleftrightarrow  (t^n_0 - t^n_1)(t^n_3 - t^n_2) \geq 0  \qquad \text{and} \qquad
 		\nu_3^2 \geq \nu_1^2 \Longleftrightarrow  (t^n_0 - t^n_2)(t^n_3 - t^n_1) \geq 0\, .
 	\end{aligned}
 \end{equation} 
 We first prove $\mu_3^2 \geq \mu_2^2\implies \nu_3^2 \geq \nu_2^2$: 
 
 \uline{When $n$ is even}: 
 \begin{itemize}
 \item[(1).] If $\mu_3^2 = \mu_2^2$, by Eq.~\eqref{mu3}, we have
 \begin{itemize}
     \item[\ding{172}] If $s^n_0 - s^n_1 = 0$, by Proposition~\ref{prop1}, we have $p_0 = p_1 = 1/2$ and $p_2 = p_3 = 0$, hence $t_0^n - t_1^n =0$ and $\nu_3^2 = \nu_2^2$.
     \item[\ding{173}] If $s_3^n = s_2^n$, we have two situations. First is $p_0>0$, then $p_2 = p_3$, by  Proposition~\ref{prop2}, we have $\nu_3^2 = \nu_2^2$ for any $p_1$. Second is $p_0 = 0$, then $s_3^n = s^n_2 = s_1^n =0$. In this case, we have $\mu_1 = \mu_2 = \mu_3$, and of course, the index $i$ for $\nu = |\nu_i|$ is equal to the index for $\mu = |\mu_i|$.
 \end{itemize}
\item[(2).] If $\mu_3^2 > \mu_2^2$, we have $p_0>0$ and $p_2 < p_3$, hence $t_0^n - t_1^n <0$ and $t_3^n - t_2^n \leq 0$ for any $p_1$, thus $\nu_3^2 \geq \nu_2^2$.
\end{itemize}
 
\uline{When $n$ is odd}:
\begin{itemize}
    \item[(1).] If $\mu_3^2 = \mu_2^2$, by Eq.~\eqref{mu3}, we have
 \begin{itemize}
     \item[\ding{172}] If $s^n_0 - s^n_1 = 0$, by Proposition~\ref{prop1}, we have $p_0 = p_1$, hence $t_0^n - t_1^n =0$ and $\nu_3^2 = \nu_2^2$.
     \item[\ding{173}] If $s_3^n = s_2^n$, we have $p_2 = p_3$, by  Proposition~\ref{prop2}, we have $\nu_3^2 = \nu_2^2$ for any $p_0$ and $p_1$.
 \end{itemize}
\item[(2).] If $\mu_3^2 > \mu_2^2$, we have 
 \begin{itemize}
     \item[\ding{174}] If $p_0>p_1$ and $p_3>p_2$, we have $t_0^n-t_1^n<0$ and $t^n_3 - t^n_2 \leq 0 $, hence $\nu_3^2 \geq \nu_2^2$.
     \item[\ding{175}] If $p_0<p_1$ and $p_3<p_2$, we have $t_0^n-t_1^n>0$ and $t^n_3 - t^n_2 \geq 0 $, hence $\nu_3^2 \geq \nu_2^2$.
 \end{itemize}
 \end{itemize}
 
By a similar argumant we can prove $\mu_3^2 \geq \mu_1^2\implies \nu_3^2 \geq \nu_1^2$ as well. Thus, the minimal output entropy pf $\Phi_\pm$ are attained in the same input state, and moreover, we have 
\begin{equation*}
    \chi(\mathcal{S}^n) = 1 -\min_{i=1,2,3} \bigg[ (1-\mathcal{P}_n) h\big(|\mu_i|\big) + \mathcal{P}_n h\big(|\nu_i|\big) \bigg] = 1 - (1-\mathcal{P}_n)h(\mu)-\mathcal{P}_n h(\nu) \geq C(\mathcal{S}^n) \geq \chi(\mathcal{S}^n),
\end{equation*}
thus the classical capacity of $\mathcal{S}^n$ is equal to its Holevo information $\chi(\mathcal{S}^n)$, as desired.
\end{proof}

\section{Necessary and sufficient conditions for zero causal gain}
In this section, we prove necessary and sufficient conditions for zero classical and coherent information causal gain of Pauli channels with forward and backward orders. Recall that in this case, the effective channel is given by
$$\mathcal{S}^n(\rho)=\big(1-\mathcal{P}_n\big) \Phi_+(\rho)  \otimes |+\rangle\langle+| \ + \  \mathcal{P}_n\Phi_-(\rho) \otimes |-\rangle\langle-|\, ,$$ 
where $\Phi_+(\rho) = \sum s_i \sigma\rho\sigma_i$ and $\Phi_-(\rho) = \sum t_i \sigma\rho\sigma_i$, and the $n$-fold composition channel is $\mathcal{N}^n(\rho) = \sum_{i=0}^{3} q_i \sigma_i \rho \sigma_i$. If $\mathcal{P}_n > 0$, since $\Phi_+ = \frac{1}{4(1-\mathcal{P}_n)} \mathcal{S}^n_+$ and $\Phi_- = \frac{1}{4\mathcal{P}_n} \mathcal{S}^n_-$, the probabilities $\{s_i\}_{i=0}^3$ of $\Phi_+$ also satisfy the statements in Proposition~\ref{prop1} and the probabilities $\{t_i\}_{i=0}^3$ of $\Phi_-$ satisfy the statements in Proposition~\ref{prop2}.
Let $\{\gamma_i\}, \{\mu_i\}, \{\nu_i\}$ be eigenvalues of $\mathcal{N}^n, \Phi_+, \Phi_-$ respectively, and $h(x) = -\frac{1+x}{2}\log \frac{1+x}{2} - \frac{1-x}{2}\log \frac{1-x}{2}$.

\subsection{The conditions for zero classical causal gain}

Since $\mathcal{N}^n(\rho) = \sum_{i=0}^{3} q_i \sigma_i \rho \sigma_i$, $\gamma_i = q_0 + 2q_i - \sum_{j=1}^3 q_j$ and $\gamma = \max_i \{|\gamma_i|\}$, the classical capacity of $\mathcal{N}^n$ is equal to its Holevo information~\cite{HolevoPauli3,HoveloPauli1,HolevoPauli2}\, i.e.,
\begin{equation}
	C(\mathcal{N}^n) = \chi(\mathcal{N}^n) = 1 - h(\gamma)\, ,
\end{equation}
and since $\mathcal{N}^n(\rho) = (1-\mathcal{P}_n)\Phi_+(\rho) + \mathcal{P}_n\Phi_-(\rho)$, we have that the causal gain $\delta_C$ associated with the classical capacity may be written as
\begin{equation}
	\begin{aligned}
		\delta_C = C(\mathcal{S}^n) - C(\mathcal{N}^n) = h(\gamma) - (1-\mathcal{P}_n)h(\mu) - \mathcal{P}_n h(\nu)\, ,
	\end{aligned}
\end{equation}
where $\mu = \max\{|\mu_i|\}$ and $\nu = \max\{|\nu_i|\}$ (recall here that $\mu_i = s_0 +2s_i - \sum_{j=1}^3 s_j$ and $\nu_i = t_0 +2t_i - \sum_{j=1}^3 t_j$ for $i=1,2,3$, so that $\gamma_i = (1-\mathcal{P}_n)\mu_i + \mathcal{P}_n \nu_i$). 

In this subsection, we  prove the following:
\begin{thm}
    The classical causal gain $\delta_C$ defined above is positive if and only if the quantity $\mathcal{P}_n > 0$, except the cases that $n$ is odd and $\mathcal{N}$ is the completely depolarizing channel.
\end{thm}
Since $h(x)$ is a strict concave function,
$\delta_C = 0$ if and only if $\gamma = \mu =\nu$ or $\mathcal{P}_n = 0$ ($\mathcal{P}_n$ is always less than 1, see discussion following Eq.~\eqref{PnPauli}). Since $\mathcal{P}_n = 0$ implies there is no causal order superposed in the quantum $\tt SWITCH$, i.e. $\mathcal{S}^n$ is just $\mathcal{N}^n$ tensor with a fixed flagged state, and the classical causal gain is zero obviously. So we only need to prove that the situation $\gamma = \mu =\nu$ if and only if $n$ is odd and $\mathcal{N}$ is the completely depolarizing channel (in this case, we have $\mathcal{P}_n > 0$ but $\delta_C = 0$).

\begin{proof}
Firstly, we prove $\gamma = \mu =\nu$ only happens if $n$ is odd and $\mathcal{N}$ is the completely depolarizing channel. As shown in Appendix C, we have $\mu$ and $\nu$ are obtained on the same index, without loss of generality, we can assume that $\mu = |\mu_3|$ and $\nu = |\nu_3|$.

We claim that $\mu = |\mu_i| = |\nu_i| = \nu$, and give the proof by considering the converse below. 

\begin{itemize}
	\item[I.]
If $|\mu_3| > \max\{|\mu_2|,|\mu_1|\}$ and $|\nu_3|\geq \max\{|\nu_2|,|\nu_1|\}$, then we can consider such two cases:
\begin{itemize}
	\item[(1)] \uline{If $\nu_3 = -\mu_3$}, we can get $|\gamma_3| = |1-2\mathcal{P}_n| \mu <\mu$, and $|\gamma_1|\leq (1-\mathcal{P}_n)|\mu_1| + \mathcal{P}_n |\nu_1| <\mu$, $|\gamma_2| <\mu$ as well, hence $\gamma <\mu$, which leads to a contradiction.
	\item[(2)] \uline{If $\nu_3 = \mu_3$}.
 
 \qquad Firstly suppose that $\mu_3>0$, so that by Proposition~\ref{prop1} we have $\mu_3>\mu_2, \mu_3>\mu_1$, which implies $p_3>p_2$ and $p_3>p_1$. 
 
\qquad On the other hand, since $\nu_3\geq \max\{\nu_1,\nu_2\}$, by Proposition~\ref{prop2}, when $n$ is even, $p_3>\max\{p_1,p_2\}$ happens only for $p_1 = p_2 = 0$, which means $\mathcal{P}_n = 0$ by Prop.~\ref{PauliPn}, which is impossible. When $n$ is odd, $p_3>\max\{p_1,p_2\}$ happens only for $p_0p_1 = 0$ and $p_0p_2=0$. If $p_0 >0$, which also implies $\mathcal{P}_n = 0$. So it is necessary that $p_0 =0$. But this contradicts $0< \mu_3^2 - \mu_2^2 = 4(s_0-s_1)(s_3-s_2)$ since $s_3>s_2$ and $s_1\geq 0$. 
 
\qquad If $\mu_3<0$, we also get a contradiction by a similar argument.
\end{itemize}
\item[II.] The cases $|\mu_3| = |\mu_2| > |\mu_1|$ and $|\mu_3| = |\mu_1| > |\mu_2|$ are ruled out for the same reasons as above.
\end{itemize}

Hence, it must hold that $\mu = |\mu_i| = |\nu_i| = \nu$. Now we analyse the sign of these eigenvalues.
\begin{itemize}
\item[\ding{172}] \uline{If $\mu_3 = \mu_2 = \mu_1$}, which implies $s_1 = s_2 = s_3$. 
 
\qquad When $n$ is odd, we have $p_1 = p_2 = p_3$, and moreover $\nu_1 = \nu_2 = \nu_3$, $t_1 = t_2 = t_3$. In this case, if $\nu_3 = -\mu_3$, by a similar argument above, we can deduce $\gamma<\mu$, which is impossible. Therefore, it is necessary that $\mu_3 = \nu_3$. While $\mu_3 = s_0 - s_3$ and $\nu_3 = t_0 - t_3$, using Proposition~\ref{prop1} and \ref{prop2}, we have $p_0 = p_i$, i.e. $\mathcal{N}$ is a completely depolarizing channel. 

\qquad When $n$ is even and $p_0>0$, we also have $p_1 = p_2 = p_3$, and $t_1 = t_2 = t_3$, do same argument above, we find $n$ must be odd, which is impossible. Now, consider $n$ is even and $p_0 =0$, in this case, $s_1 = s_2 = s_3 = 0$ and $\mu = s_0 = 1 = \nu$. Since $\sum_{i=0}^3 t_i = 1$, there must exist only one $t_i = 1$ and others are zero. If $t_0 = 1$, then $t_0 - t_1 = 1>0$ contradicts to Prop.~\ref{prop2}. If $t_0 = 0$ and $t_i = 1$, then for different $i,j,k\in\{1,2,3\}$, $t_0 - t_j = t_0 - t_k =0$ means $p_0 = p_j = p_k = 1/2$, it is impossible.

\item[\ding{173}] \uline{If $\mu_3 = \mu_2 = -\mu_1$}, which implies $s_0 = s_2 = s_3$ and $n$ is odd and $p_0 = p_2 = p_3$. By the same argument, we can obtain that $\mathcal{N}$ is a completely depolarizing channel.

\item[\ding{174}] Other cases also happen only if $n$ is odd and $\mathcal{N}$ is the completely depolarizing channel.
\end{itemize}

In summary, it is necessary that $n$ is odd and $\mathcal{N}$ is a completely depolarizing channel for $\gamma = \mu = \nu$. 

Finally, if $n$ is odd and $\mathcal{N}$ is the completely depolarizing channel, by Proposition~\ref{prop1} and \ref{prop2}, we have $s_i = t_i = 1/4$ for all $i\in\{0,1,2,3\}$, which means $\Phi_\pm$ are all completely depolarizing channels and $\gamma = \mu = \nu =0$.
\end{proof}

\subsection{The conditions for zero coherent information causal gain}

Notice that $\Phi_+(\rho) = \sum_{i=0}^3 s_i \sigma_i\rho\sigma_i$,\  $\Phi_-(\rho) = \sum_{i=0}^3 t_i \sigma_i\rho\sigma_i$, and $\mathcal{N}^n = \sum_{i=0}^3 q_i \sigma_i\rho\sigma_i = (1-\mathcal{P}_n) \Phi_+ + \mathcal{P}_n \Phi_-.$
Since the coherent information causal gain may be written as
\begin{equation}
	\delta_I = H(\vec{q}) - (1-\mathcal{P}_n)H(\vec{s}) - \mathcal{P}_n H(\vec{t})\, ,
\end{equation}
it follows that $\delta_I=0$ if and only if $\vec{s} = \vec{t}$ for the concavity of entropy, oo $\mathcal{P}_n=0$.
Here, we will prove the following result:
\begin{thm}
     The coherent information causal gain $\delta_I$ defined above is positive if and only if the quantity $\mathcal{P}_n > 0$, except the cases that $n$ is odd and $\mathcal{N}$ is the completely depolarizing channel.
\end{thm}
Similar to the classical causal gain, we only need to prove that $\vec{s} = \vec{t}$ holds if and only if $n$ is odd and $\mathcal{N}$ is the completely depolarizing channel, in which case we have $\mathcal{P}_n > 0$ but $\delta_I = 0$. Moreover, since $\Phi_+ = \frac{1}{4(1-\mathcal{P}_n)} \mathcal{S}^n_+$ and $\Phi_- = \frac{1}{4\mathcal{P}_n} \mathcal{S}^n_-$, the probability coefficients $\{s_i\}_{i=0}^3$ of $\Phi_+$ also satisfy Proposition~\ref{prop1} and the probability coefficients $\{t_i\}_{i=0}^3$ of $\Phi_-$ satisfy Proposition~\ref{prop2}.
\begin{proof}
Firstly, assume that $\vec{s}=\vec{t}$. If $n$ is even, by Proposition~\ref{prop1} and \ref{prop2}, we have $s_0 - s_i \geq 0$ and $t_0 - t_i \leq 0$. Since $\vec{s} = \vec{t}$, we have $s_0 - s_i = t_0 - t_i = 0$, which implies $p_0 = p_i = 1/2$ for $i\in\{1,2,3\}$, which is impossible.
So it is necessary that $n$ is odd. Also by Proposition~\ref{prop1}, \ref{prop2} and $s_0 - s_i = t_0 - t_i$, we have $p_0 = p_1 = p_2 = p_3 = 1/4$, that is $\mathcal{N}$ is a completely depolarizing channel. In this case, $\Phi_\pm$ are both completely depolarizing channels.

If $n$ is odd and $\mathcal{N}$ is the completely depolarizing channel, by Proposition~\ref{prop1} and \ref{prop2}, we have $s_i = t_i = 1/4$ for all $i\in\{0,1,2,3\}$, which means $\Phi_\pm$ are all completely depolarizing channels.
\end{proof}

\section{Communication enhanced for qudit depolarizing channels via quantum $\tt SWITCH$}

In this section, we show that the quantity $\mathcal{P}_n$ also determines whether the quantum $\tt SWITCH$ can enhance classical communication for qudit depolarizing channels. Given $p\in [0,1]$, let $\mathcal{D}_p$ be the depolarizing channel given by
$$\mathcal{D}_p(\rho) = (1-p)\rho + p\text{Tr}(\rho) \frac{I}{d} = (1-p) \rho + \frac{p}{d^2} \sum_{i=1}^{d^2} U_i\rho U_i^\dagger,$$
where $\{U_i\}_{i=1}^{d^2}$ are unitary operators and form an orthonormal basis of $d\times d$ matrices. Here we assume $p>0$, since no communication advantage can be achieved by the quantum $\tt SWITCH$ when $p=0$. 

Suppose that $n$ copies of the qudit depolarizing channel $\mathcal{D}_p$ are put into the quantum $\tt SWITCH$ in a superposition of forward and backward orders. Fixed the control state $\omega = |+\rangle\langle+|$, so that the effective channel is given by
\begin{equation}
	\mathcal{S}^n(\rho) = \frac{1}{2} \sum_{\pi^i,\pi^j\in \mathbf{S}} C_{\pi^i \pi^j}(\rho) \otimes |i\rangle\langle j|_C\, ,
\end{equation}
where $C_{\pi^i \pi^j}(\rho) = \sum_{s_1,\dots,s_n} C^{\pi^i(1)}_{s_{\pi^i(1)}}\cdots C^{\pi^i(n)}_{s_{\pi^i(n)}} \rho \bigg( C^{\pi^j(1)}_{s_{\pi^j(1)}}\cdots C^{\pi^j(n)}_{s_{\pi^j(n)}}\bigg)^\dagger$ and $\{C^{k}_{s_k}\}$ are Kraus operators of $k$-th $\mathcal{D}_p$.

\subsection{An exact expression of the effective channel $\mathcal{S}^n(\rho)$}
In this subsection, we will compute the exact expression of $\mathcal{S}^n$. We let $[[0,m]] := \{0,1,\dots,m\}$, and we let $\pi^1$ be the identity permutation $(1,\dots,n)$ and we let $\pi^2$ be the reversal $(n,\dots,1)$.
Notice that the Kraus operators of $\mathcal{D}_p$ are $C_0 = \sqrt{1-p} I$ and $C_i = \frac{\sqrt{p}}{d}U_i$ for $ i=1\dots,d^2$. We then have
\begin{equation}
	\begin{aligned}
		C_{\pi^1 \pi^2}(\rho) & = \sum_{s_1,\dots,s_n \in [[0,d^2]]} C^{1}_{s_{1}}\cdots C^{n}_{s_{n}} \rho  \  C^{1\dagger}_{s_{1}}\cdots C^{n\dagger}_{s_{n}} \\
		& = \sum_{k=0}^{n} \binom{n}{k} (1-p)^k \bigg(\frac{p}{d^2}\bigg)^{n-k} \sum_{s_{i_1},\dots,s_{i_{n-k}}\in[[1,d^2]]} U^{i_1}_{s_{i_1}}\cdots U^{i_{n-k}}_{s_{i_{n-k}}} \rho  \  U^{i_1\dagger}_{s_{i_1}}\cdots U^{i_{n-k}\dagger}_{s_{i_{n-k}}},
	\end{aligned}
\end{equation}
where the second equality is obtained by counting the number of $C_0$-factors in the product $C^{1}_{s_{1}}\cdots C^{n}_{s_{n}}$. Moreover, we have
\begin{equation}
	\begin{aligned}
		\sum_{s_{i_1},\dots,s_{i_t} \in [[1,d^2]]} U^{i_1}_{s_{i_1}}\cdots U^{i_t}_{s_{i_t}} \rho  \  U^{i_1\dagger}_{s_{i_1}}\cdots U^{i_t\dagger}_{s_{i_t}} & = d \sum_{s_{i_1},\dots,s_{i_{t-1}}\in [[1,d^2]]} U^{i_1}_{s_{i_1}}\cdots U^{i_{t-1}}_{s_{i_{t-1}}} \text{Tr}\bigg(\rho  \  U^{i_1\dagger}_{s_{i_1}}\cdots U^{i_{t-1}\dagger}_{s_{i_{t-1}}}\bigg) \\
		& = d^2 \sum_{s_{i_1},\dots,s_{i_{t-2}}\in [[1,d^2]]} U^{i_1}_{s_{i_1}}\cdots U^{i_{t-2}}_{s_{i_{t-2}}} \rho  \  U^{i_1\dagger}_{s_{i_1}}\cdots U^{i_{t-2}\dagger}_{s_{i_{t-2}}} \\
		& = \bigg\{
		\begin{array}{ll}
			d^t I  &\qquad \text{if $t$ is odd}\, , \\
			d^t \rho  &\qquad \text{if $t$ is even}\, ,
		\end{array} \bigg.
	\end{aligned}
\end{equation}
where the first equality follows from the fact that $\sum_{i=1}^{d^2} U_i A U_i^\dagger = d\text{Tr}(A)\ I$ for a matrix $A$, and the second equality follows from  $\sum_{i=1}^{d^2} U_i \text{Tr}(A U_i^\dagger) = d A$. 

Therefore, we have 
\begin{itemize}
	\item[(1)] \uline{If $n$ is even}, 
\begin{equation}
	\begin{aligned}
	C_{\pi^1 \pi^2}(\rho) & = \binom{n}{0} \big(\frac{p}{d}\big)^n \rho + \binom{n}{1}(1-p)^1\big(\frac{p}{d}\big)^{n-1} I +\cdots + \binom{n}{n-1}(1-p)^{n-1}\big(\frac{p}{d}\big)^{1} I + \binom{n}{n} (1-p)^n \rho \\
	& = \frac{1}{2}\bigg[(1-p+\frac{p}{d})^n - (1-p-\frac{p}{d})^n\bigg] I + \frac{1}{2}\bigg[(1-p+\frac{p}{d})^n + (1-p-\frac{p}{d})^n\bigg] \rho.
	\end{aligned}
\end{equation}
\item[(2)] \uline{If $n$ is odd}, 
\begin{equation}
	\begin{aligned}
		C_{\pi^1 \pi^2}(\rho) & = \binom{n}{0} \big(\frac{p}{d}\big)^n I + \binom{n}{1}(1-p)^1\big(\frac{p}{d}\big)^{n-1} \rho +\cdots + \binom{n}{n-1}(1-p)^{n-1}\big(\frac{p}{d}\big)^{1} I + \binom{n}{n} (1-p)^n \rho \\
		& = \frac{1}{2}\bigg[(1-p+\frac{p}{d})^n - (1-p-\frac{p}{d})^n\bigg] I + \frac{1}{2}\bigg[(1-p+\frac{p}{d})^n + (1-p-\frac{p}{d})^n\bigg] \rho.
	\end{aligned}
\end{equation}
\end{itemize}
Altogether, we have 
\begin{equation}
  B:=	C_{\pi^1 \pi^2}(\rho) = \frac{1}{2}\bigg[(1-p+\frac{p}{d})^n - (1-p-\frac{p}{d})^n\bigg] I + \frac{1}{2}\bigg[(1-p+\frac{p}{d})^n + (1-p-\frac{p}{d})^n\bigg] \rho.
\end{equation}
Here we compute $C_{\pi^1 \pi^2}(\rho)$ for $\pi^1$ is identity permutation and $\pi^2$ is the reversal permutation, in which case we find that $C_{\pi^1 \pi^2}(\rho)$ is linear about $\rho$ and $I$. Indeed, for any two permutations $\pi^1$ and $\pi^2$, as they can also be regarded as two orders, a similar conclusion can be obtained, see~\cite{permutation}.

Hence, the effective channel may be written in matrix form as
\begin{equation}
	\mathcal{S}^n(\rho) = \frac{1}{2} 
	\begin{pmatrix}
		(1-p)^n\rho + [1-(1-p)^n] \frac{I}{d} & B \\
		B & (1-p)^n\rho + [1-(1-p)^n] \frac{I}{d}
	\end{pmatrix},
\end{equation}
and moreover, in such a case the quantity $\mathcal{P}_n$ may be given by
\begin{equation}\label{p_n}
	\mathcal{P}_n = \text{Tr}[(I\otimes |-\rangle\langle-|)\mathcal{S}^n(\rho)] = \frac{1}{2}-\frac{1}{4}\bigg((d+1)(1-p+\frac{p}{d})^n - (d-1)(1-p-\frac{p}{d})^n\bigg),
\end{equation}
which is independent to $\rho$. By direct calculation one may also show that $\mathcal{P}_n$ increases monotonically with $p$, and that $0\leq \mathcal{P}_n \leq \frac{1}{2} - \frac{(d+1)+(-1)^{n+1}(d-1)}{4d^n} \leq \frac{1}{2}$.

As such, we can rewrite the effective channel as
\begin{equation}
	\mathcal{S}^n(\rho) = (1-\mathcal{P}_n)\Phi_+ (\rho) \otimes |+\rangle\langle+| + \mathcal{P}_n \Phi_-(\rho) \otimes |-\rangle\langle-|\, ,
\end{equation}
where $\Phi_\pm$ are two trace-preserving linear maps given by
\begin{equation}
	\Phi_+(\rho) = (1-\lambda_1) \rho + \lambda_1 \frac{I}{d} \qquad \text{and} \qquad \Phi_-(\rho) = (1-\lambda_2) \rho + \lambda_2 \frac{I}{d},
\end{equation}
where $\lambda_1 = \frac{1-(1-p)^n+\frac{d}{2}[(1-p+\frac{p}{d})^n-(1-p-\frac{p}{d})^n]}{1+\frac{1}{2}[(d+1)(1-p+\frac{p}{d})^n-(d-1)(1-p-\frac{p}{d})^n)]}$ and $\lambda_2 = \frac{1-(1-p)^n-\frac{d}{2}[(1-p+\frac{p}{d})^n-(1-p-\frac{p}{d})^n]}{1-\frac{1}{2}[(d+1)(1-p+\frac{p}{d})^n-(d-1)(1-p-\frac{p}{d})^n)]}$. 

Notice that $0\leq \lambda_1 \leq 1$ and $1\leq \lambda_2 \leq \frac{d^2}{d^2-1}$, so that $\Phi_\pm$ are two depolarizing channels.

\subsection{The classical causal gain for $\mathcal{D}_p$}
In this subsection, we will show that $\mathcal{P}_n$ also determines whether quantum $\tt SWITCH$ can enhance classical communication for qudit depolarizing channels $\mathcal{D}_p(\rho) = (1-p)\rho + p \frac{I}{d}$. 

Firstly, as mentioned above, the effective channel may be written as 
\[
\mathcal{S}^n(\rho) = (1-\mathcal{P}_n) \Phi_+(\rho)\otimes|+\rangle\langle+| + \mathcal{P}_n \Phi_-(\rho) \otimes |-\rangle\langle-|\, ,
\]
where 
\begin{equation}
	\Phi_+(\rho) = \big(1-\lambda_1\big)\rho + \lambda_1\frac{I}{d}\qquad \text{and} \qquad \Phi_-(\rho) = \big(1-\lambda_2\big)\rho + \lambda_2\frac{I}{d}\, ,
\end{equation}
and in such a case the quantity $\mathcal{P}_n$ defined in Eq.~\eqref{p_n} is zero if and only if $p=0$.

Since $\Phi_\pm$ are both depolarizing channels (in fact, $\Phi_-$ is entanglement-breaking), their Holevo information is additive~\cite{King2003,Shor2002}, so that $\chi(\Phi_\pm \otimes \mathcal{M}) = \chi(\mathcal{M}) + \chi(\Phi_\pm)$ for an arbitrary channel $\mathcal{M}$. Moreover, as $\Phi_\pm$ are both covariant with respect to the group $SU(d)$, the Holevo information $\chi(\Phi_{\pm})$ may be written as~\cite{holevo2002remarks}
\begin{equation}
	\chi(\Phi_\pm) = \log d - H^{\text{min}}(\Phi_\pm)\, ,
\end{equation}
where $H^{\text{min}}(\Phi_\pm)$ is the minimal output entropy of $\Phi_\pm$. In addition, by the concavity of entropy, the minimal output entropy $H^{\text{min}}(\Phi_\pm)$ are all attained on an arbitrary pure state.

Therefore, using the same argument as subsection 1 of Appendix C, the classical information of $\mathcal{S}^n$ is bounded as
\begin{equation}\label{bound}
	C(\mathcal{S}^n) \leq \log d - (1-\mathcal{P}_n) H^{\text{min}}(\Phi_+) - \mathcal{P}_n H^{\text{min}}(\Phi_-) = (1-\mathcal{P}_n) C(\Phi_+) + \mathcal{P}_n C(\Phi_-)\, ,
\end{equation}
the upper bound is achieved by an ensemble of orthogonal pure states $\{|i\rangle\langle i|\}_{i=1}^d$ with uniform probabilities $p_i = 1/d$, i.e.,
\begin{equation}
	H\big(\mathcal{S}^n(|i\rangle\langle i|)\big) = (1-\mathcal{P}_n) H^{\text{min}}(\Phi_+) + \mathcal{P}_n H^{\text{min}}(\Phi_-) + H(\mathcal{P}_n)\, .
\end{equation}
Now since
\begin{equation}
	\mathcal{S}^n(\sum_{i=1}^{d} p_i|i\rangle\langle i|) = \frac{I}{d}\otimes \bigg((1-\mathcal{P}_n)|+\rangle\langle+| + \mathcal{P}_n|-\rangle\langle -|\bigg)\, ,
\end{equation}
it follows that
\begin{equation}
	H\bigg(\mathcal{S}^n(\sum_{i=1}^{d} p_i|i\rangle\langle i|)\bigg) = \log d +  H(\mathcal{P}_n)\, ,
\end{equation}
thus the Holevo information of $\mathcal{S}^n$ is no less than the bound given by .~\eqref{bound}. As classical capacity is always no less than Holevo information, we obtain
\begin{equation}
	\begin{aligned}
		C(\mathcal{S}^n) & =  \log d - (1-\mathcal{P}_n) H^{\text{min}}(\Phi_+) - \mathcal{P}_n H^{\text{min}}(\Phi_-) \\
		& =  \log d \\ 
        & + \bigg[\frac{1}{2}+\frac{1}{4}\bigg((d+1)(1-p+\frac{p}{d})^n - (d-1)(1-p-\frac{p}{d})^n\bigg)\bigg]\bigg[\big(1-\lambda_1 + \frac{\lambda_1}{d}\big)\log\big(1-\lambda_1 + \frac{\lambda_1}{d}\big) + (d-1)\frac{\lambda_1}{d}\log (\frac{\lambda_1}{d})\bigg] \\
		& + \bigg[\frac{1}{2}-\frac{1}{4}\bigg((d+1)(1-p+\frac{p}{d})^n - (d-1)(1-p-\frac{p}{d})^n\bigg)\bigg] \bigg[\big(1-\lambda_2 + \frac{\lambda_2}{d}\big)\log\big(1-\lambda_2 + \frac{\lambda_2}{d}\big) + (d-1)\frac{\lambda_2}{d}\log (\frac{\lambda_2}{d})\bigg]\, .
	\end{aligned}
\end{equation}
On the another hand, since $\mathcal{D}_p^{n}(\rho) := (\mathcal{D}_p \circ \cdots \circ \mathcal{D}_p)(\rho) = (1-p)^n\rho + \bigg(1-(1-p)^n\bigg) \frac{I}{d}$, it follows that
\begin{equation}
	C(\mathcal{D}_p^n) = \log d + \bigg((1-p)^n+\frac{1-(1-p)^n}{d}\bigg)\log\bigg((1-p)^n+\frac{1-(1-p)^n}{d}\bigg) + (d-1)\frac{1-(1-p)^n}{d}\log\bigg(\frac{1-(1-p)^n}{d}\bigg)\, ,
\end{equation}
thus
\begin{equation}
	\delta_C = C(\mathcal{S}^n) - C(\mathcal{D}_p^n) \geq 0\, .
\end{equation}
After some simple but lengthy calculations, we arrive at the following:
\begin{thm}
	The classical causal gain $\delta_C$ for qudit depolarizing channels is equal to zero if and only if either (i) $\mathcal{P}_n = 0$, i.e. $p=0$. or (ii) $p=1$ i.e. $\mathcal{D}$ is completely depolarizing and $n$ is odd.
\end{thm}
\begin{figure}
    \centering
    \includegraphics[width=0.49\textwidth]{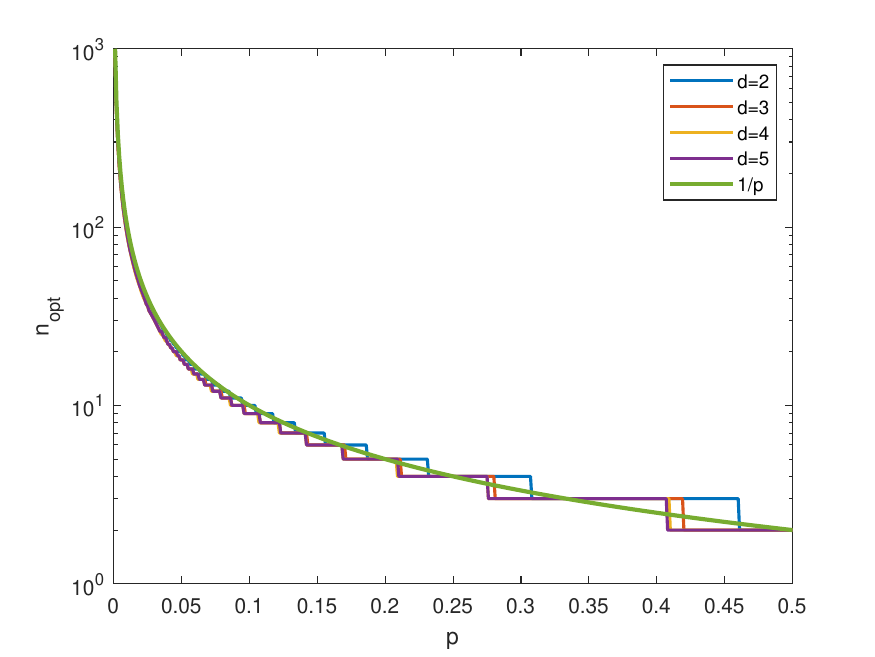}
    \caption{The optimal number of channels for communication enhancement via quantum $\tt SWITCH$. The $y$ axis $n_{opt}$ is the optimal number of channels which is approximately inversely proportional to low error probability $p<0.5$. For larger $p\geq 0.5$, $n_{opt}$ is always 2.  }
    \label{Fig5}
\end{figure}
\begin{figure*}
	\subfigure[d=2]{\includegraphics[width=0.4\textwidth]{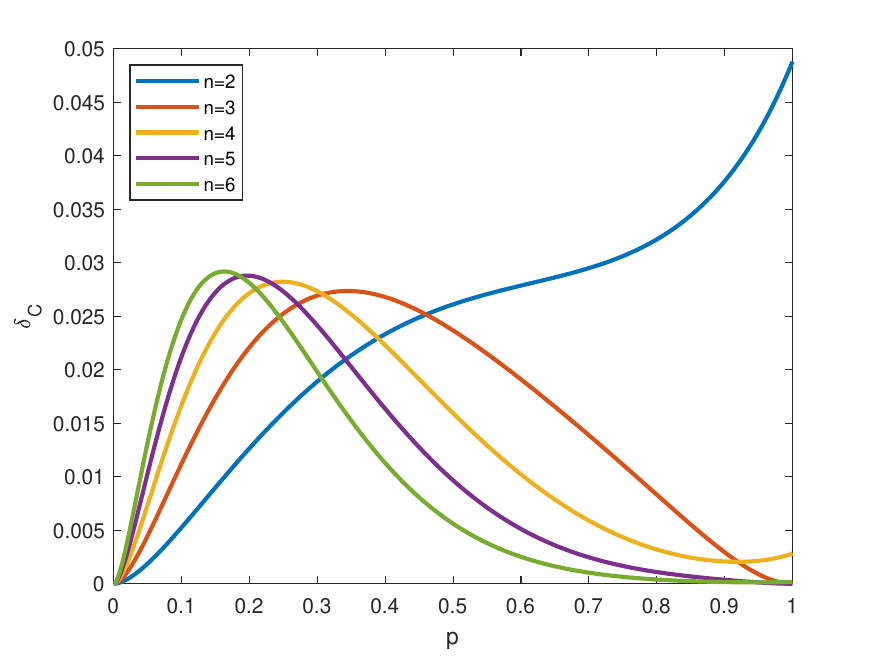}}
	\subfigure[d=3]{\includegraphics[width=0.4\textwidth]{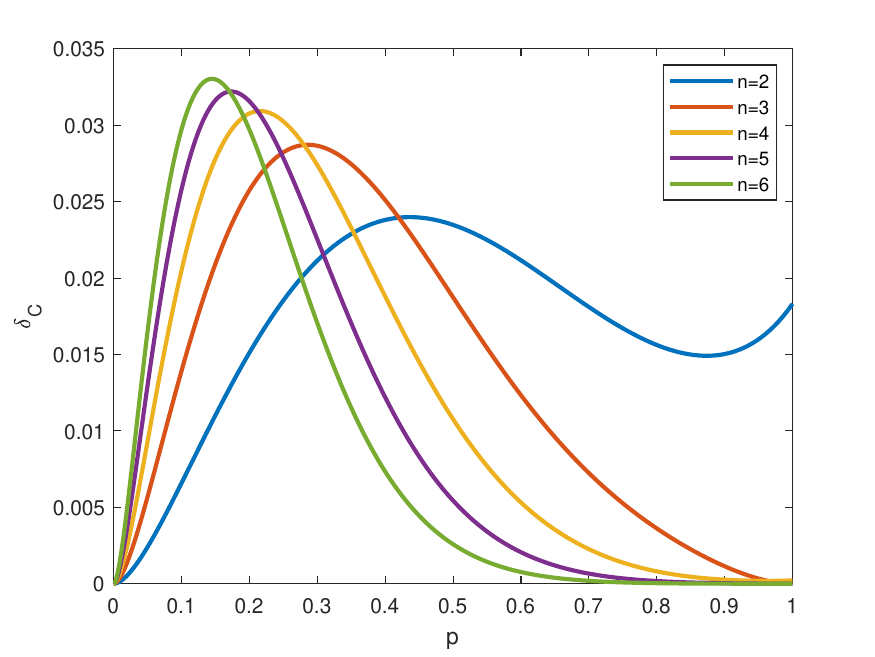}}
        \subfigure[p=0.1]{\includegraphics[width=0.4\textwidth]{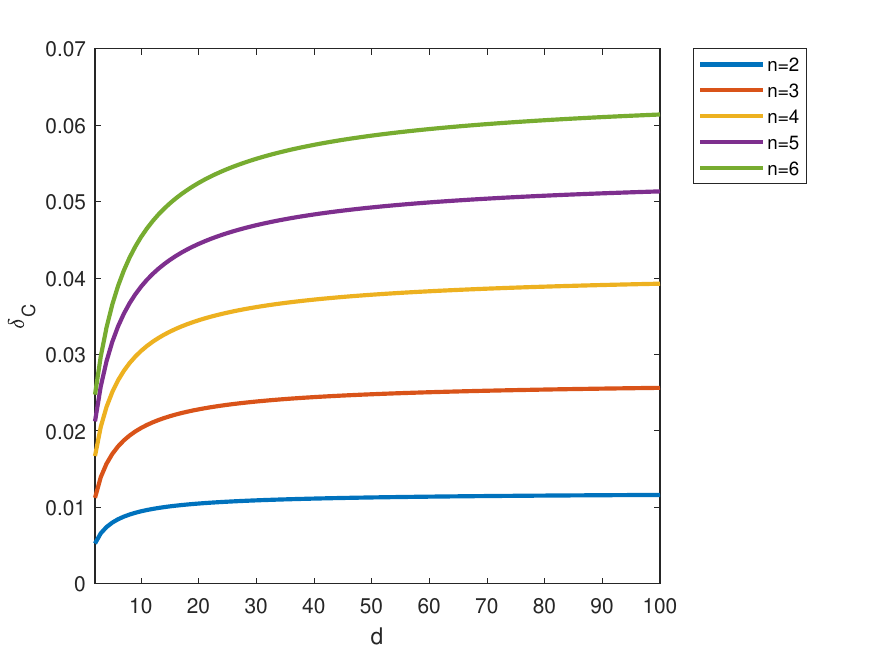}}
	\subfigure[p=0.5]{\includegraphics[width=0.4\textwidth]{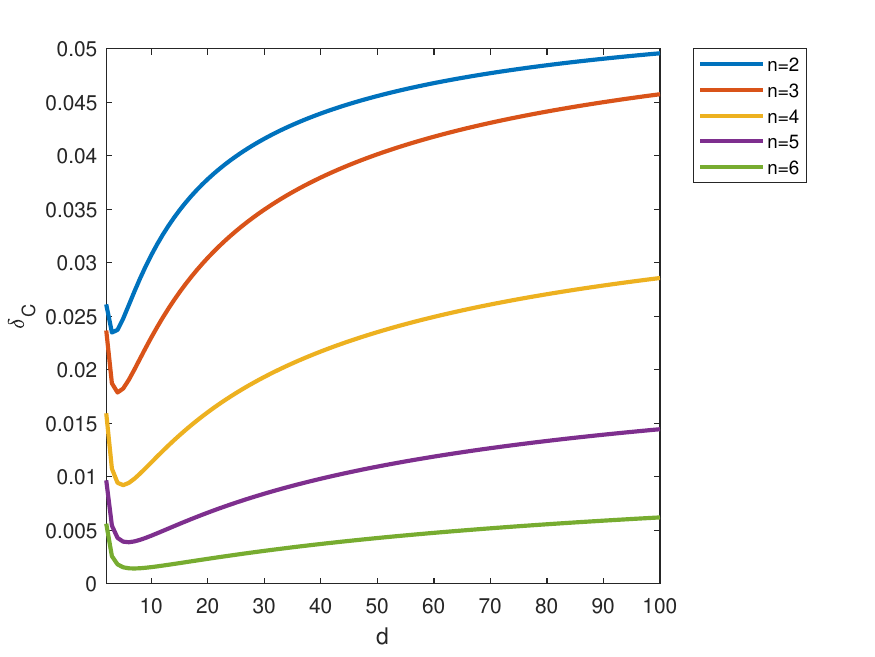}}
	\caption{{\label{Fig6}} The classical causal gain of depolarizing channels. (a) and (b): The $x$ axis is noise probability $p$ and $y$ axis is classical causal gains. (a) is $\delta_C$ For qubit depolarizing channels. (b) is $\delta_C$ For qutrit depolarizing channels. (c) and (d): The classical causal gain for different noise areas. (c) shows $\delta_C$ with the number of channels and dimensions of noise probability $p=0.1$. (d) shows the classical causal gain of middle noise probability decreases with $n$, but is increasing with $d>5$.}
\end{figure*}

We note that this theorem takes the  same form as Theorem~\ref{DeltaCX}, showing $\mathcal{P}_n$ is also useful in arbitrary dimension. In Fig.~\ref{Fig6} we plot classical causal gains. In particular, we find that for low noise areas, the classical causal gain is increasing with respect to $d$. In Fig.~\ref{Fig5} we plot the optimal number of channels via quantum $\tt SWITCH$, which approximately satisfies $p\cdot n_{opt} \approx 1$. Moreover, for middle noise areas, it is decreasing with $n$ but increasing for $d>5$.

\end{widetext}

\end{document}